\RequirePackage{amsmath}
\documentclass[number]{elsarticle}
\usepackage{url}
\usepackage{hyperref}
\usepackage{multicol}
\usepackage{amsfonts}
\usepackage{amssymb}
\usepackage{graphics,graphicx}
\usepackage{epic}
\usepackage{epsfig,float}
\usepackage{verbatim}
\usepackage{pdfsync}
\usepackage{gastex}
\usepackage{bbm}

\usepackage{xcolor}

\usepackage[utf8]{inputenc}\usepackage[T1]{fontenc}% Polskie literki

\renewcommand{\le}{\leqslant}
\renewcommand{\ge}{\geqslant}

\newcommand{\eps}{\varepsilon}
\newcommand{\emp}{\emptyset}

\newcommand{\Sig}{\Sigma}
\newcommand{\sig}{\sigma}
\newcommand{\noin}{\noindent}

\newcommand{\bi}{\begin{itemize}}
\newcommand{\ei}{\end{itemize}}
\newcommand{\be}{\begin{enumerate}}
\newcommand{\ee}{\end{enumerate}}
\newcommand{\bd}{\begin{description}}
\newcommand{\ed}{\end{description}}
\newcommand{\bq}{\begin{quote}}
\newcommand{\eq}{\end{quote}}

\newcommand{\co}{\colon}
\newcommand{\id}{\mathbbm{1}}

\newcommand{\cA}{{\mathcal A}}
\newcommand{\cB}{{\mathcal B}}
\newcommand{\cC}{{\mathcal C}}
\newcommand{\cD}{{\mathcal D}}

\newcommand{\cN}{{\mathcal N}}

\newcommand{\cP}{{\mathcal P}}

\newcommand{\cS}{{\mathcal S}}

\newcommand{\e}{\emph}
\newcommand{\sto}{\twoheadrightarrow}

\def\shu{\mathbin{\mathchoice
{\rule{.3pt}{1ex}\rule{.3em}{.3pt}\rule{.3pt}{1ex}
\rule{.3em}{.3pt}\rule{.3pt}{1ex}}
{\rule{.3pt}{1ex}\rule{.3em}{.3pt}\rule{.3pt}{1ex}
\rule{.3em}{.3pt}\rule{.3pt}{1ex}}
{\rule{.2pt}{.7ex}\rule{.2em}{.2pt}\rule{.2pt}{.7ex}
\rule{.2em}{.2pt}\rule{.2pt}{.7ex}}
{\rule{.3pt}{1ex}\rule{.3em}{.3pt}\rule{.3pt}{1ex}
\rule{.3em}{.3pt}\rule{.3pt}{1ex}}\mkern2mu}}

\newtheorem{theorem}{Theorem}
\newtheorem{lemma}[theorem]{Lemma}
\newtheorem{proposition}[theorem]{Proposition}

\newdefinition{definition}[theorem]{Definition}

\newdefinition{remark}[theorem]{Remark}
\newdefinition{example}[theorem]{Example}

\newproof{proof}{Proof}

\begin{document}

\title{State Complexity of Pattern Matching in Regular Languages\tnoteref{thanks}}

\tnotetext[thanks]{This work was supported by the Natural Sciences and Engineering Research Council of Canada 
grant No.~OGP0000871.}

\author[dcs]{Janusz A. Brzozowski}
\ead{brzozo@uwaterloo.ca}
\author[dpm]{Sylvie Davies\corref{cor}}
\ead{sldavies@uwaterloo.ca}
\author[dcs]{Abhishek Madan}
\ead{a7madan@edu.uwaterloo.ca}
\address[dcs]{David R. Cheriton School of Computer Science, University of Waterloo, Waterloo, ON, Canada N2L 3G1}
\address[dpm]{Department of Pure Mathematics, University of Waterloo, Waterloo, ON, Canada N2L 3G1}
\cortext[cor]{Corresponding author}

\begin{abstract}
In a simple pattern matching problem one has a pattern $w$ and a text $t$, which are words over a finite alphabet $\Sig$. 
One may ask whether $w$ occurs in $t$, and if so, 
where?
More generally, we may have a set $P$ of patterns and a set $T$ of texts, where $P$ and $T$ are regular languages.
We are interested 
whether any word of $T$ begins with a word of $P$, ends with a word of $P$, has a word of $P$ as a 
factor,
or has a word of $P$ as a subsequence. 
Thus we are interested in the languages $(P\Sig^*)\cap T$, 
$(\Sig^*P)\cap T$, $(\Sig^* P\Sig^*)\cap T$, and $(\Sig^*\shu P)\cap T$, where $\shu$ is the shuffle operation.
The state complexity $\kappa(L)$ of a regular language $L$ is the number of states in the minimal deterministic finite automaton recognizing $L$. 
We derive the following upper bounds on the state complexities of our pattern-matching languages, where 
$\kappa(P)\le m$, and $\kappa(T)\le n$:
$\kappa((P\Sig^*)\cap T) \le mn$; 
$\kappa((\Sig^*P)\cap T) \le 2^{m-1}n$;
$\kappa((\Sig^*P\Sig^*)\cap T) \le (2^{m-2}+1)n$; 
and $\kappa((\Sig^*\shu P)\cap T) \le (2^{m-2}+1)n$. 
We prove that these bounds are tight, and that 
to meet them, the alphabet must have at least two letters
in the first three cases,  and  at least $m-1$ letters in the last case.
We also consider the special case where $P$ is a single word $w$, and obtain the following tight upper bounds:
$\kappa((w\Sig^*)\cap T) \le m+n-1$;
$\kappa((\Sig^*w)\cap T) \le (m-1)n-(m-2)$;
$\kappa((\Sig^*w\Sig^*)\cap T) \le (m-1)n$; 
and $\kappa((\Sig^*\shu w)\cap T) \le (m-1)n$. 
For unary languages, we have a tight upper bound of $m+n-2$ in all eight of the aforementioned cases. 
 \end{abstract}


\maketitle

\section{Introduction}

Given a regularity-preserving operation on regular languages, we may ask the following natural question: in the worst case, how many states are necessary and sufficient for a deterministic finite automaton (DFA) to accept the language resulting from the operation, in terms of the number of states of the input DFAs? 
For example, consider the intersection of two languages: if the input DFAs have $m$ and $n$ states respectively, then an $mn$-state DFA is sufficient to accept the 
intersection; this follows 
by the usual direct product construction. 
It was proved by Yu, Zhuang, and Salomaa~\cite{YZS94} that an $mn$-state DFA is also necessary in the worst case; for all $m,n \ge 1$, there exist 
a language accepted by an $m$-state DFA and a language accepted by an $n$-state DFA whose intersection is accepted by a minimal DFA with $mn$ states.

This worst-case value is called the \e{state complexity}~\cite{GMRY16,Mas70,YZS94} of the operation.
The \e{state complexity} of a regular language $L$, denoted by $\kappa(L)$, is the number of states in the minimal DFA accepting $L$. Thus $\kappa(L) = n$ means the minimal DFA for $L$ has exactly $n$ states, and $\kappa(L) \le n$ means $L$ can be recognized by an $n$-state DFA.
If a language has state complexity $n$, we indicate this by the subscript $n$ and use $L_n$ instead of $L$.
Then the state complexity of an operation is the worst-case state complexity of the result of the operation, expressed in terms of the maximal allowed state complexity of the inputs. For example, the state complexity of intersection is $mn$ because if 
$\kappa(K) \le m$ and $\kappa(L) \le n$, then $\kappa(K \cap L) \le mn$ and this bound is tight for all $m,n \ge 1$.

Aside from ``basic'' operations like union, intersection, concatenation and star, the state complexity of \e{combined operations}~\cite{SSY07} such as ``star of intersection'' and ``star of union'' has also been studied. We investigate the state complexity of
new combined operations inspired by pattern matching problems.

For a comprehensive treatment of pattern matching, see~\cite{CrHa97}.
In a pattern matching problem we have a \emph{text} and a \emph{pattern}. 
In its simplest form, the pattern $w$ and the text $t$ are both words over an alphabet $\Sig$.
Some natural questions about patterns in texts include the following: Does $w$ 
occur in $t$, 
and if so, where?

Pattern matching has many applications. 
Aho and Corasick~\cite{AhCo75} developed an algorithm to determine all  occurrences of words from a finite pattern in a given text; this algorithm leads to significant improvements in the speed of bibliographic searches.
Pattern matching is used in bioinformatics~\cite{EIWA15};
 in this context  the text $t$ is often a DNA sequence, and the pattern $w$ is a sequence of nucleotides searched for in the text.

More generally, we can have a \e{set} $P$ of patterns and a \e{set} $T$ of texts. These could be finite sets, or they could be arbitrary regular languages, specified by a finite automaton or a regular expression.
For example, many text editors and text processing utilities have a \e{regular expression search} feature, which finds all lines in a text file that match a certain regular expression. 
In this context, the pattern set $P$ is often a regular language (but not always, as software implementations of ``regular expressions'' typically have extra features allowing 
irregular 
languages to be specified).
We can view a text file as either an ordered sequence of single-word texts $t$ (each representing a line of the file), or if the order of lines is not important, as a finite set $T$.
There could also be cases where it is useful to allow $T$ to be an arbitrary regular language rather than a finite set; for example, $T$ could be the set of all possible interleaved execution traces from the processes in a distributed system, as described in~\cite{DAW88}.

In this paper, we ask whether a pattern from the set $P$ occurs as a \e{prefix}, \e{suffix}, \e{factor} or \e{subsequence} of a text from the set $T$.
If $u,v,w\in\Sigma^*$ and $w=uv$, then $u$ is a {\em prefix\/} of $w$ and $v$ is a {\em suffix\/} of $w$. 
If $w=xvy$ for some $v,x,y\in\Sigma^*$, then $v$ is a {\em factor\/} of $w$.
%Note that a prefix or suffix of $w$ is also a factor of $w$.
If $w=w_0a_1w_1\cdots a_nw_n$, where $a_1,\ldots,a_n\in \Sigma$, and $w_0,\ldots,w_n\in\Sigma^*$, then $v=a_1\cdots a_n$ is a \e{subsequence}
of $w$. 
%Note that every factor of $w$ is a subsequence of $w$.

If $L$ is any  language, then $L \Sig^*$ is the \emph{right ideal} generated by $L$,
$\Sig^* L$ is the \emph{left ideal} generated by $L$, and 
$\Sig^*L\Sig^* $ is the \emph{two-sided ideal} generated by $L$.

The \emph{shuffle} $u\shu v$ of words $u,v\in\Sig^*$ is defined as follows:
$$ u\shu v= \{u_1v_1\cdots u_kv_k \mid u=u_1\cdots u_k,  v=v_1\cdots v_k,
u_1,\ldots,u_k,v_1,\ldots, v_k\in \Sig^*\}.$$
The shuffle of two languages $K$ and $L$ over $\Sig$ is defined by
$$ K\shu L=\bigcup_{u \in K, v \in L} u\shu v.$$
The language $\Sig^* \shu L$ is an \emph{all-sided ideal}.
The language $\Sig^*\shu w$ consists of all words that contain $w$ as a subsequence.
Such a language could be used, for example, to determine whether a report has all the required sections and that they are in the correct order.

The combined operations we consider are of the form
 ``the intersection of $T$ with the right (left, two-sided, all-sided) ideal generated by $P$''.
We study four problems with pattern sets $P\subseteq \Sig^*$ and text sets $T\subseteq \Sig^*$. 
\be
\item
Find $(P\Sig^*) \cap T$, the set of all the words in $T$ each of which begins with a word in $P$.
\item
Find $(\Sig^*P) \cap T$, the set of all the words in $T$ each of which ends with a word in $P$.
\item
Find $(\Sig^* P\Sig^*) \cap T$ the set of all the words in $T$ each of which has a word in $P$ as a factor.
\item
Find $(\Sig^* \shu P) \cap T$, the set of all the words in $T$ each of which has a word of $P$ as a subsequence.
\ee
We then repeat these four problems for the case where the pattern is a single word $w$.
In all eight cases we find the state complexity of these operations.
We show that for languages $P$, $T$ and $\{w\}$ such that $\kappa(P)\le m$, $\kappa(T) \le n$, $\kappa(\{w\})\le m$,
the following upper bounds hold:
\be
\item
General case:
\be
\item
Prefix: $\kappa((P\Sig^*)\cap T) \le mn$. 
\item
Suffix: $\kappa((\Sig^*P)\cap T)\le 2^{m-1}n$. 
\item
Factor: $\kappa((\Sig^*P\Sig^*)\cap T)   \le (2^{m-2}+1)n$. 
\item
Subsequence:
$\kappa((\Sig^*\shu P)\cap T) \le  (2^{m-2}+1)n$. 
\ee
\item
Single-word case:
\be
\item
Prefix:
$\kappa((w\Sig^*)\cap T) \le m+n-1$.
\item
Suffix: $\kappa((\Sig^*w)\cap T) \le (m-1)n-(m-2)$.
\item
Factor: $\kappa((\Sig^*w\Sig^*)\cap T) \le (m-1)n$.
\item
Subsequence: $\kappa((\Sig^*\shu w)\cap T) \le (m-1)n$. 
\ee
\ee
Moreover, in each case there exist languages $P_m$, $T_n$, $\{w\}_m$ that meet the upper bounds.

In the general prefix, suffix and factor cases, there exist binary witnesses meeting the bounds.
For the general subsequence case, an alphabet of at least $m-1$ letters is needed to reach this bound.
For the single-word cases we use binary witnesses to reach each of the bounds.

In Section~\ref{sec:prefixword}, we consider prefix matching in the case where $P = \{w\}$ is a single word.
In addition to considering arbitrary alphabets, in that section we also look at the 
 case where $P$ and $T$ are languages over a \e{unary alphabet}.
We prove a tight upper bound of $m+n-2$ on the state complexity of $w\Sig^* \cap T$ in the unary alphabet case.
It turns out that when $P$ and $T$ are unary languages, the single-word prefix matching case coincides with all the other cases.
Thus none of the upper bounds from above can be reached with unary alphabets.
See Remark \ref{rmk:unary} in Section \ref{sec:prefixword} for more details.

\section{Terminology and Notation}

A \emph{deterministic finite automaton (DFA)} is a
5-tuple
$\cD=(Q, \Sigma, \delta, q_0,F)$, where
$Q$ is a finite non-empty set of \emph{states},
$\Sig$ is a finite non-empty \emph{alphabet},
$\delta\colon Q\times \Sig\to Q$ is the \emph{transition function},
$q_0\in Q$ is the \emph{initial} state, and
$F\subseteq Q$ is the set of \emph{final} states.
%We extend $\delta$ to functions $\delta\colon Q\times \Sig^*\to Q$ and $\delta\colon 2^Q\times \Sig^*\to 2^Q$ as usual.
We extend $\delta$ to a function $\delta \co Q \times \Sig^* \to Q$ inductively as follows:
for $q \in Q$, define $\delta(q,\eps) = q$, and for $x \in \Sig^*$ and $a \in \Sig$, define $\delta(q,xa) = \delta(\delta(q,x),a)$.
We extend it further to $\delta \co 2^Q \times \Sig^* \to 2^Q$ by setting $\delta(S,w) = \{ \delta(q,w) \mid q \in S\}$ for $S \subseteq Q$.
A~DFA $\cD$ \emph{accepts} a word $w \in \Sigma^*$ if ${\delta}(q_0,w)\in F$. The language accepted by $\cD$ 
is the set of all words that $\cD$ accepts, and
is denoted by $L(\cD)$. If $q$ is a state of $\cD$, then the language $L_q(\cD)$ of $q$ is the language accepted by the DFA $(Q,\Sigma,\delta,q,F)$. 
A state is \emph{empty} (or \emph{dead} or a \emph{sink state}) if its language is empty. Two states $p$ and $q$ of $\cD$ are 
\emph{indistinguishable} 
 if $L_p(\cD) = L_q(\cD)$. 
A state $q$ is \emph{reachable} if there exists $w\in\Sig^*$ such that $\delta(q_0,w)=q$.
A DFA $\cD$ is \emph{minimal} if it has the smallest number of states  and the smallest alphabet  among all DFAs accepting $L(\cD)$.
It is well known that a DFA is minimal if it uses the smallest alphabet, all of its states are reachable, and no two states are 
indistinguishable.
Two DFAs are \e{isomorphic} if (informally) the only difference between them is the names assigned to the states.

A \e{nondeterministic finite automaton (NFA)} is a 
5-tuple 
$\cN = (Q,\Sigma,\delta,q_0,F)$, where $\delta$ is now a function $\delta \co Q \times \Sig \to 2^Q$, and all other components are as in a DFA. Extending $\delta$ to a function $\delta \co 2^Q \times \Sig^* \to 2^Q$, the NFA $\cN$ accepts a word $w \in \Sig^*$ if $\delta(\{q_0\},w) \cap F \ne \emp$. As with DFAs, the language accepted by the NFA $\cN$ is the set of all accepted words.

Let $L$ be a language over $\Sig$. The \e{quotient} of $L$ by a word $x \in \Sig^*$ is the set $x^{-1}L = \{y \in \Sig^* \mid xy \in L\}$. In a DFA $\cD = (Q,\Sig,\delta,q_0,F)$, if $\delta(q_0,w) = q$, then $L_q(\cD) = w^{-1}L(\cD)$.

%Let $\cT_{Q}$ be the set of all $n^n$ transformations of $Q$; then $\cT_{Q}$ is a monoid under composition. 
%The \emph{in-degree} of a state $q$ in a transformation $t$ is the cardinality of the set $\{p \mid pt=q\}$.

A \e{transformation} of a set $Q$ is a function $t \co Q \to Q$.
The \emph{image} of $q\in Q$ under the transformation $t$ is denoted by $qt$.
If $s,t$ are transformations of $Q$, their composition is denoted by $st$ and defined by
$q(st)=(qs)t$; that is, composition is performed from \e{left to right}.
The \e{preimage} of $q \in Q$ under the transformation $t$ is denoted by $qt^{-1}$, and is defined to be the set $qt^{-1} = \{ p \in Q \mid pt = q\}$.
This notation extends to sets: for $S \subseteq Q$, we have $St = \{ qt \mid q \in S\}$ and $St^{-1} = \{ p \in Q \mid pt \in S\}$.

For $k\ge 2$, a transformation $t$ of a set $P=\{q_0,q_1,\ldots,q_{k-1}\} \subseteq  Q $ is a \emph{$k$-cycle}
if $q_0t=q_1$, $q_1t=q_2$, \dots, $q_{k-2}t=q_{k-1}$, 
$q_{k-1}t=q_0$, 
and $qt = q$ for all $q \in Q \setminus P$.
This $k$-cycle is denoted by  $(q_0,q_1,\ldots,q_{k-1})$. 
A~2-cycle $(q_0,q_1)$ is a \emph{transposition}.
%A transformation  that sends all the states of $P$ to $q$ and acts as the identity on the remaining states is denoted by $(P \to q)$.  If $P=\{p\}$ we write  $(p\to q)$ for $(\{p\} \to q)$.
 The identity transformation 
of $Q$ is denoted by $\id$; while this notation omits the set $Q$, it can generally be inferred from context.
If $Q$ is a set of natural numbers (e.g., $Q = \{0,1,\dotsc,n-1\}$ for some $n$),
the notation $(_i^j \; q\to q+1)$ denotes a transformation that sends $q$ to $q+1$ for $i\le q\le j$ and is the identity for the remaining elements
of $Q$,
and  $(_i^j \; q\to q-1)$ is defined similarly.

In a DFA $\cD=(Q, \Sigma, \delta, q_0,F)$, each letter $a\in \Sig$ induces a transformation of the set of states $Q$, defined by $q \mapsto \delta(q,a)$ for $q \in Q$.
We denote this transformation by $\delta_a$.
Specifying the transformation $\delta_a$ induced by each letter $a \in \Sig$ completely specifies the transition function $\delta$, so we often define $\delta$ in this way.
We write $a \co t$ to mean $\delta_a = t$; for example, 
if $Q = \{0,1,\dotsc,n-1\}$, then
$a \co (0,1,\dotsc,n-1)$ means the transformation induced by $a$ in the DFA $\cD$ is the cycle $(0,1,\dotsc,n-1)$.
We extend the $\delta_a$ notation from letters to words: if $w = a_1\dotsb a_k$ for $a_1,\dotsc,a_k \in \Sig$, then $\delta_w = \delta_{a_1} \dotsb \delta_{a_k}$.

A \emph{dialect} of a  regular language $L$ is a language obtained from $L$
 by replacing or deleting letters of $\Sigma$ in the words of $L$.
 In this paper we use only dialects obtained by 
permuting the letters of $\Sig$.
 Thus, for example, if $L(a,b) =b^*(aab  \cup a)$, then $L(b,a)=a^*(bba \cup b)$.
 The notion of a dialect is also extended to DFAs. 

Henceforth we sometimes refer to state complexity as simply \e{complexity}, since we do not discuss other measures of complexity in this paper.
 
%More precisely, for an alphabet $\Sigma'$ and a partial map $\pi \colon \Sigma \mapsto \Sigma'$,
%we obtain a dialect of $L$ by replacing each letter $a \in \Sigma$ by $\pi(a)$ in every word of $L$,
%or deleting the word entirely if $\pi(a)$ is undefined.
%We write $L(\pi(a_1),\dots, \pi(a_k))$ to denote the dialect of $L(a_1,\dots,a_k)$ given by $\pi$,
%and we denote undefined values of $\pi$ by  ``$-$''.
%For example, if $L(a,b,c)= \{a, ab, ac\}$ then its dialect $L(b,-,d)$ is the language $\{b, bd\}$.
%Undefined values for letters at the end of the alphabet are omitted; thus, for example, 
%if $\Sigma =\{a,b,c,d,e\}$, $\pi(a)=b$, $\pi(b)=a$, $\pi(c)=c$ and $\pi(d)=\pi(e)=-$, we write $L(b,a,c)$ for $L(b,a,c,-,-)$.

\section{Prefix Matching}

Let $T$ and $P$ be regular languages over an alphabet $\Sig$.
We compute the set 
$L$ of all the words of $T$ that are prefixed by words in $P$; that is, the language $L = \{ wx \mid w \in P, wx \in T \}=(P\Sig^*)\cap T$.
We want to find the worst-case state complexity of $L$.
\begin{theorem}
 For $m,n\ge 1$, 
%if $P$ and $T$ are of state complexities $\le m$ and $\le n$, respectively, 
if $\kappa(P) \le m$ and $\kappa(T) \le n$,
 then $\kappa((P\Sig^*)\cap T) \le mn$, and this bound is tight if the cardinality of $\Sig$ is at least 2.
 \end{theorem}
 \begin{proof}
 The language $P\Sig^*$ is the right ideal generated by $P$. It is known that the state complexity of $P\Sig^*$ is at most $m$~\cite{BJL13}.
 Furthermore, it was shown in~\cite{YZS94} that the complexity of 
intersection is at most $mn$,
if the first input has complexity at most $m$ and the second has complexity at most $n$. 
Hence $mn$ is an upper bound on the complexity of $(P\Sig^*)\cap T$.

Next we find witnesses that meet this bound.
%; they are restrictions to the alphabet $\{a,b\}$ of the most complex regular language introduced in~\cite{Brz13}.
Let $T_n(a,b)$ be accepted by the DFA $\mathcal{D}_n(a,b) = (Q_n, \Sigma, \delta_T, 0, \{ n-1 \})$, where $Q_n=\{0,1,\dots,n-1\}$, $\Sig=\{a,b\}$ and $\delta_T$ is defined by the transformations $a \colon (0, 1, \ldots, n-1)$, and 
$b\colon \id$; 
see Figure~\ref{fig:regL}.
This DFA is minimal because the shortest word in $a^*$ accepted by state $q$ is $a^{n-1-q}$; this shortest word distinguishes $q$ from any other state.
\begin{figure}[ht!]
\unitlength 8.5pt
\begin{center}\begin{picture}(37,6)(0,4)
\gasset{Nh=1.8,Nw=3.8,Nmr=1.25,ELdist=0.4,loopdiam=1.5}
	{\small
\node(0)(1,7){0}\imark(0)
\node(1)(8,7){1}
\node(2)(15,7){2}
\node[Nframe=n](3dots)(22,7){$\dots$}
\node(n-2)(29,7){$n-2$}
\node(n-1)(36,7){$n-1$}\rmark(n-1)
\drawloop(0){$b$}
\drawedge(0,1){$a$}
\drawloop(1){$b$}
\drawedge(1,2){$a$}
\drawloop(2){$b$}
\drawedge(2,3dots){$a$}
\drawedge(3dots,n-2){$a$}
\drawloop(n-2){$b$}
\drawedge(n-2,n-1){$a$}
\drawedge[curvedepth= 3.0,ELdist=-1.0](n-1,0){$a$}
\drawloop(n-1){$b$}
	}
\end{picture}\end{center}
\caption{Minimal DFA $\cD_n(a,b)$ of $T_n(a,b)$.}
\label{fig:regL}
\end{figure}

Now let $P_m=P_m(a,b)=T_m(b,a)$ be the dialect of $T_m(a,b)$ with the roles of $a$ and $b$ interchanged.
Thus the DFA
$\mathcal{D}_m(b,a) = (Q_m, \Sigma, \delta_P, 0, \{ m-1 \})$, where 
$Q_m=\{0,1,\dots, m-1\}$ and $\delta_P$ is defined by 
$a\colon \id$, 
$b \colon (0, 1, \ldots, m-1)$, is the minimal DFA of $T_m(b,a)$.
This DFA is minimal because any state is distinguished from any other state by the shortest word in 
$b^*$ 
that it accepts.

To find $P_m\Sigma^*$, we concatenate the language $P_m$ with the language $\Sig^*$.
Note that once state $m-1$ is reached in the DFA $\cD'_m(b,a)$ recognizing $P_m\Sig^*$, every word is accepted. Thus the transition from $m-1$ to 0 is not needed, because it is replaced by a self-loop on state $m-1$ under $b$. Thus we obtain the DFA 
$\mathcal{D}'_m(b,a) = (Q_m, \Sigma, \delta_P', 0, \{ m-1 \})$ of Figure~\ref{fig:regMStar}, where $\delta_P'$ is defined by $a \colon \id$, $b \colon (_0^{m-2} q \rightarrow q + 1)$.

\begin{figure}[ht!]
\unitlength 8.5pt
\begin{center}\begin{picture}(37,8)(0,4)
\gasset{Nh=1.8,Nw=3.8,Nmr=1.25,ELdist=0.4,loopdiam=1.5}
	{\small
\node(0)(1,7){0}\imark(0)
\node(1)(8,7){1}
\node(2)(15,7){2}
\node[Nframe=n](3dots)(22,7){$\dots$}
\node(m-2)(29,7){$m-2$}
\node(m-1)(36,7){$m-1$}\rmark(n-1)
\drawloop(0){$a$}
\drawedge(0,1){$b$}
\drawloop(1){$a$}
\drawedge(1,2){$b$}
\drawloop(2){$a$}
\drawedge(2,3dots){$b$}
\drawedge(3dots,m-2){$b$}
\drawloop(m-2){$a$}
\drawedge(m-2,m-1){$b$}
\drawloop(m-1){$a,b$}
	}
\end{picture}\end{center}
\caption{Minimal DFA $\cD'_m(b,a)$ of 
$P_m\Sigma^*$.
}
\label{fig:regMStar}
\end{figure}
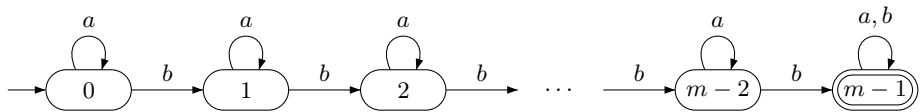

Our last task is to find a DFA accepting $(P_m\Sig^*)\cap T_n$ and prove that it is minimal and has $mn$ states.
To achieve this we find the direct product $\cD_L$ of $\cD'_m(b,a)$ and $\cD_n(a,b)$; an example of this product for 
$m=n=4$ 
is given in Figure~\ref{fig:exK}.
Let $\mathcal{D}_L=(Q_m \times Q_n, \Sigma, \delta_L, (0, 0), \{ (m-1, n-1) \})$, where $\delta_L((p, q), a) = (\delta_P'(p, a), \delta_T(q, a))$. 
%An example $\mathcal{D}_K$ is shown in \autoref{fig:exK}.
Since DFA $\cD_L$ has $mn$ states, it remains to prove that every state if $\cD_L$ is reachable and every two states are distinguishable.

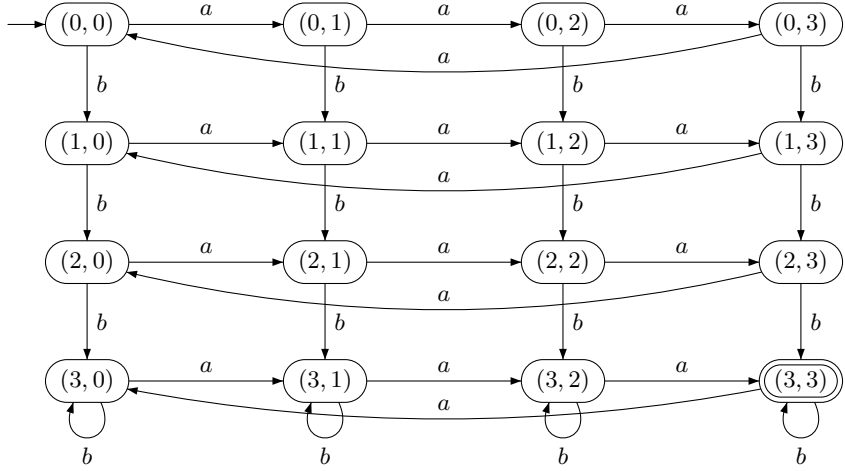
\begin{figure}[th]
\unitlength 9pt
\begin{center}\begin{picture}(35,20)(0,-3)
\gasset{Nh=1.8,Nw=3.5,Nmr=1.25,ELdist=0.4,loopdiam=1.5}
	{\small
\node(0'0)(2,15){$(0,0)$}\imark(0'0)
\node(1'0)(2,10){$(1,0)$}
\node(2'0)(2,5){$(2,0)$}
\node(3'0)(2,0){$(3,0)$}

\node(0'1)(12,15){$(0,1)$}
\node(1'1)(12,10){$(1,1)$}
\node(2'1)(12,5){$(2,1)$}
\node(3'1)(12,0){$(3,1)$}

\node(0'2)(22,15){$(0,2)$}
\node(1'2)(22,10){$(1,2)$}
\node(2'2)(22,5){$(2,2)$}
\node(3'2)(22,0){$(3,2)$}

\node(0'3)(32,15){$(0,3)$}
\node(1'3)(32,10){$(1,3)$}
\node(2'3)(32,5){$(2,3)$}
\node(3'3)(32,0){$(3,3)$}\rmark(3'3)

\drawedge(0'0,0'1){$a$}
\drawedge(0'1,0'2){$a$}
\drawedge(0'2,0'3){$a$}
\drawedge[curvedepth=2.0,ELside=r](0'3,0'0){$a$}

\drawedge(1'0,1'1){$a$}
\drawedge(1'1,1'2){$a$}
\drawedge(1'2,1'3){$a$}
\drawedge[curvedepth=2.0,ELside=r](1'3,1'0){$a$}

\drawedge(2'0,2'1){$a$}
\drawedge(2'1,2'2){$a$}
\drawedge(2'2,2'3){$a$}
\drawedge[curvedepth=2.0,ELside=r](2'3,2'0){$a$}

\drawedge(3'0,3'1){$a$}
\drawedge(3'1,3'2){$a$}
\drawedge(3'2,3'3){$a$}
\drawedge[curvedepth=1.6,ELside=r](3'3,3'0){$a$}

\drawedge(0'0,1'0){$b$}
\drawedge(1'0,2'0){$b$}
\drawedge(2'0,3'0){$b$}
\drawloop[loopangle=-90](3'0){$b$}

\drawedge(0'1,1'1){$b$}
\drawedge(1'1,2'1){$b$}
\drawedge(2'1,3'1){$b$}
\drawloop[loopangle=-90](3'1){$b$}

\drawedge(0'2,1'2){$b$}
\drawedge(1'2,2'2){$b$}
\drawedge(2'2,3'2){$b$}
\drawloop[loopangle=-90](3'2){$b$}

\drawedge(0'3,1'3){$b$}
\drawedge(1'3,2'3){$b$}
\drawedge(2'3,3'3){$b$}
\drawloop[loopangle=-90](3'3){$b$}

}
\end{picture}\end{center}
\caption{The direct product of 
$\mathcal{D}'_{4}(b,a)$ and $\mathcal{D}_{4}(a,b)$
for intersection.
}
\label{fig:exK}
\end{figure}

We observe that $\delta_L((0,0),a^q b^p) = (p,q)$, for all $0 \le p \le m-1$ and $0 \le q \le n-1$. Therefore, every state is reachable. We also observe that the minimal word
in $a^*b^*$
accepted by a state $(p,q)$ is $a^{n-1-q} b^{m-1-p}$, where $0 \le p \le m-1$ and $0 \le q \le n-1$. Therefore, each state in $Q_m \times Q_n$ has a unique minimal word
in $a^*b^*$; 
this makes all states pairwise distinguishable. Hence, $\cD_L$ is minimal, and has state complexity $mn$.
\qed
\end{proof}

\section{Suffix Matching}
Let $T$ and $P$ be regular languages over an alphabet $\Sig$.
We are now interested in 
the worst-case state complexity of 
the set $L$ of all the words of $T$ that end with words in $P$. More formally, $L = \{ xw \mid w \in P, xw \in T \}=(\Sig^*P)\cap T$.

\begin{proposition}
For $m,n \ge 2$,
%if $P_m$ and $T_n$ are of statrte complexities $\le m$ and $\le n$, respectively, 
if $\kappa(P) \le m$ and $\kappa(T) \le n$,
 then $\kappa((\Sig^*P)\cap T) \le 2^{m-1}n$.
% , and this bound is tight if the cardinality of $\Sig$ is at least 2.
\label{prop:suffix}
 \end{proposition}
 \begin{proof}
 The language $\Sig^*P$ is the left ideal generated by $P$. It is known that the state complexity of this ideal is at most $2^{m-1}$~\cite{BJL13}.
 Furthermore, the complexity of intersection is at most the product of the state complexities of the two operands. Hence $2^{m-1}n$ is an upper bound on the complexity of $(\Sig^*P)\cap T$.
\qed
\end{proof}

Our next goal is to prove that this upper bound is tight. We describe witnesses 
$P_m$ and $T_n$
that meet the upper bound.
%Our DFA  is a restriction of the most complex DFA used in~\cite{BrDa17a} to a binary alphabet.
Let $T_n(a,b)$ be accepted by the DFA $\mathcal{D}_n(a,b) = (Q_n, \Sigma, \delta_T, 0, \{ n-1 \})$, where $Q_n=\{0,1,\dots,n-1\}$, $\Sig=\{a,b\}$ and $\delta_T$ is defined by the transformations $a \colon (0, 1, \ldots, n-1)$, $b \colon (1,2,\dots,n-1)$. 
See Figure~\ref{fig:Lsuffix}.
This DFA is minimal because the shortest word in $a^*$ accepted by state $q$ is $a^{n-1-q}$; this shortest word distinguishes $q$ from any other state.

\begin{figure}[ht!]
\unitlength 8.5pt
\begin{center}\begin{picture}(37,7)(0,3)
\gasset{Nh=1.8,Nw=3.8,Nmr=1.25,ELdist=0.4,loopdiam=1.5}
	{\small
\node(0)(1,7){0}\imark(0)
\node(1)(8,7){1}
\node(2)(15,7){2}
\node[Nframe=n](3dots)(22,7){$\dots$}
\node(n-2)(29,7){$n-2$}
\node(n-1)(36,7){$n-1$}\rmark(n-1)
\drawedge(0,1){$a$}
\drawedge(1,2){$a,b$}
\drawloop(0){$b$}
\drawedge(2,3dots){$a,b$}
\drawedge(3dots,n-2){$a,b$}
\drawedge(n-2,n-1){$a,b$}
\drawedge[curvedepth= 4.0,ELdist=-1.0](n-1,0){$a$}
\drawedge[curvedepth= 2.0,ELdist=-1.0](n-1,1){$b$}
	}
\end{picture}\end{center}
\caption{Minimal DFA $\cD_n(a,b)$ of $T_n(a,b)$.}
\label{fig:Lsuffix}
\end{figure}

It turns out 
that $\cD_n(a,b)$, and its dialect $\cD_m(b,a)$ shown in Figure~\ref{fig:Msuffix}, act as  witnesses 
 in the case of suffix matching.
We denote the language of 
the DFA $\cD_m(b,a)$
by $P_m(b,a)$.
Let $E = \Sig^{m-2}(a\Sig^{m-2})^*$.
Then 
%the language accepted by $\cD_m(b,a)$ is denoted 
$P_m(b,a)$ can be described 
by the regular expression 
$ 
(a \cup bEb)^* bE.
$
Now we have 
$$
\Sig^*P_m(b,a)=\Sig^*(a \cup bEb)^* bE = \Sig^*bE =\Sig^*b\Sig^{m-2}(a\Sig^{m-2})^*.$$
The new generator of the left ideal $\Sig^*P_m(b,a)$ is 
$G_m=b\Sig^{m-2}(a\Sig^{m-2})^*$.
It consists of any word of length $m-1$ beginning with  $b$, possibly followed by any number of words of length $m-1$ beginning with $a$.
An NFA accepting the left ideal is shown in Figure~\ref{fig:Nsuffix}.

%%\newpage
%%Let $M_m=L_m(b,a)$; then
% Let $\mathcal{D}_m(b,a) =$ 
%Let $\mathcal{D}_m(b,a) = (Q_m, \Sigma, \delta_M, 0, \{ m-1 \})$, where 
%%$Q_m=\{0,1,\dots, m-1\}$ and 
%$\delta_M$ is defined by  
%$a \colon (1,2,\dots,n-1)$, 
%$b \colon (0, 1, \ldots, m-1)$. 
%See Figure~\ref{fig:Msuffix}.
%DFA $\cD_m(b,a)$ is minimal
%because any state is distinguished from any other state by the shortest word in $b^*$ that it accepts.
% Let $P_m(b,a)$ be the language accepted by $\cD_m(b,a)$; notice that $P_m(b,a)$ is in fact $D_m(a,b)$.
%lt

\begin{figure}[ht!]
\unitlength 8.5pt
\begin{center}\begin{picture}(37,7)(0,3)
\gasset{Nh=1.8,Nw=3.8,Nmr=1.25,ELdist=0.4,loopdiam=1.5}
	{\small
\node(0)(1,7){0}\imark(0)
\node(1)(8,7){1}
\node(2)(15,7){2}
\node[Nframe=n](3dots)(22,7){$\dots$}
\node(m-2)(29,7){$m-2$}
\node(m-1)(36,7){$m-1$}\rmark(m-1)
\drawedge(0,1){$b$}
\drawedge(1,2){$a,b$}
\drawloop(0){$a$}
\drawedge(2,3dots){$a,b$}
\drawedge(3dots,m-2){$a,b$}
\drawedge(m-2,m-1){$a,b$}
\drawedge[curvedepth= 4.0,ELdist=-1.0](m-1,0){$b$}
\drawedge[curvedepth= 2.0,ELdist=-1.0](m-1,1){$a$}
	}
\end{picture}\end{center}
\caption{Minimal DFA $\cD_m(b,a)$ of $P_m(b,a)$.}
\label{fig:Msuffix}
\end{figure}

%\newpage

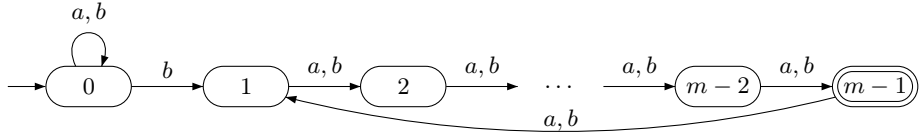
\begin{figure}[ht!]
\unitlength 8.5pt
\begin{center}\begin{picture}(37,7)(0,3)
\gasset{Nh=1.8,Nw=3.8,Nmr=1.25,ELdist=0.4,loopdiam=1.5}
	{\small
\node(0)(1,7){0}\imark(0)
\node(1)(8,7){1}
\node(2)(15,7){2}
\node[Nframe=n](3dots)(22,7){$\dots$}
\node(m-2)(29,7){$m-2$}
\node(m-1)(36,7){$m-1$}\rmark(m-1)
\drawedge(0,1){$b$}
\drawedge(1,2){$a,b$}
\drawloop(0){$a,b$}
\drawedge(2,3dots){$a,b$}
\drawedge(3dots,m-2){$a,b$}
\drawedge(m-2,m-1){$a,b$}
\drawedge[curvedepth= 2.0,ELdist=-1.0](m-1,1){$a,b$}
	}
\end{picture}\end{center}
\caption{NFA for $\Sig^*P_m(b,a)$.}
\label{fig:Nsuffix}
\end{figure}

Before we prove that the bound of Proposition~\ref{prop:suffix} is tight, we need  
a different characterization of 
the language $\Sig^*P_m(b,a)$. 
To describe a DFA for this language,
we will use binary $(m-1)$-tuples which we denote by $x =(x_1,\dots, x_{m-1})$.
\begin{definition}
Define the following DFA:
$$\cB_m(b,a)=( \{0,1\}^{m-1},\{a,b\},(0,\dots,0),\beta,\{ x\in \{0,1\}^{m-1} \mid x_1=1 \} ),$$
where 
\[ 
\beta((x_1,x_2,\dots,x_{m-2},x_{m-1}), \sigma)= 
\begin{cases}
(x_2,x_3,\dots, x_{m-1},x_1), & \text{if } \sigma = a;\\
(x_2,x_3,\dots, x_{m-1},1), & \text{if } \sigma = b.
\end{cases}
\]
In other words, the input $\sig=a$ shifts the tuple $x$ one position to the left cyclically, 
while $\sig = b$ shifts the tuple to the left, losing the first component and replacing $x_{m-1}$ by 1.
DFA $\cB_4(b,a)$ is shown in Figure~\ref{fig:B4}.
%\newpage

\begin{figure}[ht]
\unitlength 8.5pt
\begin{center}
  \begin{picture}(35,20)(-3,2)
    \gasset{Nh=1.8,Nw=4.5,Nmr=1.25,ELdist=0.4,loopdiam=1.5}
    {\footnotesize
    %\node(000)(0,15){$(0,0,0)$}\imark(000)
    \node(001)(8,12){$(0,0,1)$}
    \node(010)(16,17){$(0,1,0)$}
    \node(100)(24,20){$(1,0,0)$}\rmark(100)
    \node(000)(0,12){$(0,0,0)$}\imark(000)
    \node(101)(24,14){$(1,0,1)$}\rmark(101)
    \node(011)(16,7){$(0,1,1)$}
    \node(110)(24,10){$(1,1,0)$}\rmark(110)
    \node(111)(24,4){$(1,1,1)$}\rmark(111)
    \drawloop(000){$a$}
    \drawloop(111){$a,b$}
  \drawedge(000,001){$b$}
 \drawedge(001,010){$a$}
\drawedge(010,100){$a$}
\drawedge[curvedepth= 4,ELdist=.4](110,101){$a,b$}  
\drawedge(001,011){$b$}
\drawedge(010,101){$b$}
%    \drawedge(101,011){$a,b$}
\drawedge(011,111){$b$}
\drawedge(011,110){$a$}

\drawedge[curvedepth= -4,ELdist=-1.8](100,001){$a,b$}  
\drawedge[curvedepth= -3.5,ELdist=-1.8](101,011){$a,b$}
%    \drawedge[ELdist=-2](110,101){$a,b$}
    }
  \end{picture}
\end{center}
\caption{The DFA $\cB_4(b,a)$ for $\Sig^*P_4(b,a)$.}
\label{fig:B4}
\end{figure}

%\begin{figure}[ht]
%\unitlength 8.5pt
%\begin{center}
%  \begin{picture}(35,16)(-7,2)
%    \gasset{Nh=1.8,Nw=4.5,Nmr=1.25,ELdist=0.4,loopdiam=1.5}
%    {\footnotesize
%    \node(000)(0,15){$(0,0,0)$}\imark(000)
%    \node(001)(8,15){$(0,0,1)$}
%    \node(010)(16,15){$(0,1,0)$}
%    \node(100)(24,15){$(1,0,0)$}\rmark(100)
%    \node(111)(0,9){$(1,1,1)$}\rmark(111)
%    \node(011)(8,9){$(0,1,1)$}
%    \node(101)(16,9){$(1,0,1)$}\rmark(101) 
%    \node(110)(12,3){$(1,1,0)$}\rmark(110)
%    
%    \drawloop(000){$a$}
%    \drawloop(111){$a,b$}
%    
%    \drawedge(000,001){$b$}
%    \drawedge(001,010){$a$}
%    \drawedge(010,100){$a$}
%    \drawedge[curvedepth= -3,ELdist=.4](100,001){$a,b$}  
%    
%    \drawedge(001,011){$b$}
%    \drawedge(010,101){$b$}
%    \drawedge(101,011){$a,b$}
%    \drawedge(011,111){$b$}
%    \drawedge(011,110){$a$}
%    \drawedge[ELdist=-2](110,101){$a,b$}
%    }
%  \end{picture}
%\end{center}
%\caption{The DFA $\cB_4$ for $\Sig^*L_4(b,a)$.}
%\label{fig:L2}
%\end{figure}
\end{definition}

\begin{proposition}
\label{prop:reachsuffix}
All the states of $\cB_m(b,a)$ are reachable and pairwise distinguishable.
\end{proposition}
\begin{proof}
Consider a state $(x_1,\dots,x_{m-1})$, and view it as the binary representation of a number $k$.
State $k=0$ is reachable by $\eps$ and $k=1$ by $b$.
If $k>1$ is even, it is reachable from $k/2$ by $a$, and $k+1$ is reachable from $k/2$ by $b$.
Thus all the tuples in $\{0,1\}^{m-1}$ are reachable.

We note that if a state $q = (x_1,\ldots,x_i,\ldots,x_{m-1})$ has $x_i = 1$, then $q$
accepts the word $a^{i-1}$. 
For each state $q$, define $\mathbf{A}(q) = \{ a^{i-1} \mid x_i = 1 \}$.
Since each state has a unique binary representation, each state has a unique $\mathbf{A}(q)$, which is a subset of all words accepted by $q$.
Therefore, if $p$ and $q$ are distinct 
states, they are pairwise distinguishable by words in $\mathbf{A}(p) \cup \mathbf{A}(q)$. 
\qed 
\end{proof}

In the example of Figure~\ref{fig:B4}, we have 
$\mathbf{A}(001)=\{aa\}$,
$\mathbf{A}(010)=\{a\}$,
$\mathbf{A}(011)=\{a,aa\}$,
$\mathbf{A}(100)=\{\eps\}$,
$\mathbf{A}(101)=\{\eps,aa\}$,
$\mathbf{A}(110)=\{\eps,a\}$,
$\mathbf{A}(111)=\{\eps,a,aa\}$.

%\begin{remark}
%Let $|q|_b$ be the number of $b$s in $q$. If $\beta(q,x)=p$ for some $x\in \Sig^*$, then $|q|_b \le |p|_b$.
%\end{remark}
%\begin{proof}
%Clearly if $x=0$, then $|q|_1 = |p|_1$.
%If $x=1$, then $|q|_1 = |p|_1$ if $x_1=1$, and $|q|_1 < |p|_1$ if $x_1=0$.
%\qed
%\end{proof}

%\begin{remark}
%If $x_1=1$, then $\beta(x,a)=\beta(x,b)$.
%\end{remark}

Recall that the left ideal $\Sig^* P_m(b,a)$ is generated by the language $G_m=b\Sig^{m-2}(a\Sig^{m-2})^*$, that is, $\Sig^* P_m(b,a) = \Sig^* G_m$.
\begin{lemma}
\label{lem:suffix-iso}
DFA $\cB_m(b,a)$ is isomorphic to the minimal DFA of $\Sig^*P_m(b,a)$.
\end{lemma}
\begin{proof}
First we prove that 
each state $(x_1,\dotsc,x_{m-1})$ of $\cB_m(b,a)$ accepts $G_m$, and thus $\cB_m(b,a)$ accepts a superset of 
$\Sig^*G_m = \Sig^*P_m(b,a)$.
Let $w$ be an arbitary word from $G_m$.
Since $w$ begins with $b$, this letter ``loads'' a 1 into position $x_{m-1}$. 
Then this $b$ is followed by $m-2$ arbitrary letters, which 
shift
the 1 into position $x_1$.
If there is no more input, the word $w$ is accepted.
Otherwise, the next letter is an $a$.
This shifts the positions left cyclically, moving the 1 from position $x_1$ back into position $x_{m-1}$.
Following the $a$, we have $m-2$ arbitrary letters, which 
shift
the 1 to position $x_1$.
If there is no more input, the word $w$ is accepted; otherwise the next letter must be an $a$, and the behaviour just described repeats until there is no more input.
This shows that $G_m$ is accepted from every state.
Thus $\Sig^*G_m \subseteq L(\cB_m(b,a))$.

Next we prove that $L(\cB_m(b,a)) \subseteq \Sig^*G_m$.
If $w\in L(\cB_m
(b,a))
$, then $w$ has length at least $m-1$. 
% and the state $q = (q_1,\dotsc,q_{m-1})$ reached by $w$ has $q_1 = 1$.
Let $w = \sig_1\sig_2 \dotsb \sig_k$, where $\sig_i \in \Sig$.
Consider the prefix $\sig_1 \dotsb \sig_{k-(m-2)}$ of $w$. 
%We know that $p_{m-1} = q_1$, since from state $p$, applying $m-2$ letters shifts the element in position $p_{m-1}$ to the first position.
If $\sig_{k-(m-2)} = b$, then $w$ 
is in
$\Sig^*b\Sig^{m-2} \subseteq \Sig^*G_m$ and we are done.

If $\sig_{k-(m-2)} = a$, then $w$ 
is in
$\Sig^*a\Sig^{m-2}$.
Now our proof strategy is as follows: jump back $m-1$ letters and look at $\sig_{k-(m-2)-(m-1)}$.
If this letter is a $b$, then $w$ 
is in
$\Sig^*b\Sig^{m-2}a\Sig^{m-2} \subseteq \Sig^*G_m$ and we are done.
If it's an $a$, then $w$
is in
$\Sig^*(a\Sig^{m-2})^2$, and we can keep jumping back $m-1$ letters at a time until we find a $b$.

More formally, we claim there exists $\ell \ge 0$ such that $\sig_{k-(m-2)-\ell(m-1)} = b$, and for $0 \le i < \ell$ we have $\sig_{k-(m-2)-i(m-1)} = a$; thus 
$w$ is in
$\Sig^*b\Sig^{m-2}(a\Sig^{m-2})^\ell$, and we are done.

To see this, suppose the above claim is false.
We can write $k-(m-2) = \ell(m-1) + j$, where $\ell$ is the quotient upon dividing $k-(m-2)$ by $m-1$, and $j$ is the remainder with $0 \le j < m-1$.
Since the claim is false, we have $\sig_j = \sig_{k-(m-2)-\ell(m-1)} = a$. 
In fact, we have $\sig_{k-(m-2)-i(m-1)} = a$ for $0 \le i \le \ell$.
It follows that $w$ 
is in
 $\sig_1 \dotsb \sig_{j-1}(a\Sig^{m-2})^{\ell+1}$.
Since $j-1 < m-1$, the prefix $\sig_1 \dotsb \sig_{j-1}$ cannot lead to an accepting state.
Now, if we are in a non-accepting state, and we apply a word from the language $(a\Sig^{m-2})^*$, we will remain in a non-accepting state.
Thus $w$ is not accepted, which is a contradiction.
So the claim must be true, and this completes the proof. 
\qed
\end{proof}

To finally prove that $(\Sig^*P_m(b,a)) \cap T_n(a,b)$ meets the bound $2^{m-1}n$, we construct the direct product of
the DFAs $\cB_m(b,a)$ and $\cD_n(a,b)$.
       %, where the states of $\cB_m(b,a)$ are represented by the decimal equivalent of the binary $m-1$ tuples.  
%An example for such a direct product is given in Figure~\ref{fig:suffix_cross} for $m=4$ and $n=5$.
We show that all $2^{m-1}n$ states in the direct product are reachable and pairwise distinguishable.

%To prove reachability, we will use the following lemma:
We will use the following lemma in the proof of reachability:
\begin{lemma}
\label{lem:prodreach}
If 
(a) 
DFAs $\cB = (P, \Sig, p_0, \beta, G)$ and $\cD = (Q, \Sig, q_0, \delta, F)$ are minimal DFAs, 
(b) 
 $\delta_\sig$
is bijective on $Q$ for all $\sig \in \Sig$, and 
(c) 
every state in $\{ p_0 \} \times Q$ is reachable in the direct
product of the DFAs $\cP = \cB \times \cD$, then every state in $P \times Q$ is reachable in $\cP$.
\end{lemma}
\begin{proof}
Suppose every state in $\{p_0\} \times Q$ is reachable. 
We will show that $(p,q)$ is reachable for all $p \in P$ and $q \in Q$. 
Let $w$ be a word over $\Sigma$ that such that 
$p_0\beta_w = p$; 
such a word exists since $\cB$ is minimal.
Since $\delta_\sig$ is bijective for all $\sig \in \Sigma$, 
the transformation $\delta_w$ is bijective and hence has an inverse.
So we may reach $(p,q)$ by first reaching 
$(p_0,q\delta_w^{-1})$ and then applying $\delta_w$.
%First, we will show that, if every state in $\{ p \} \times Q$ is reachable for some $p \in P$,
%and $\beta(p,\sig) = p^\prime$ for some $\sig \in \Sig$,
%then every state in $\{ p^\prime \} \times Q$ is reachable.
%For an arbitrary $q^\prime \in Q$, we know that, by the bijectivity of $\delta_\sig$, there is a $q \in Q$ where $\delta_\sig(q) = q^\prime$.
%Since $(p, q)$ is reachable, $(p^\prime, q^\prime)$ must also be reachable.
%As a result, every state in $\{ p^\prime \} \times Q$ is reachable.
%Using the fact that every state in $P$ is reachable from $p_0$ due to $\cB_m$ being minimal,
%and by repeatedly applying the above observation starting from $\{ p_0 \} \times Q$ (which we know is reachable),
%we have shown that every state in $P \times Q$ is reachable in $\cP$.
\qed
\end{proof}

We can now prove the following theorem:
\begin{theorem}
 For $m,n\ge 2$,
%if $P_m$ and $T_n$ are of state complexities $\le m$ and $\le n$, respectively, 
if $\kappa(P) \le m$ and $\kappa(T) \le n$,
 then $\kappa((\Sig^*P)\cap T) \le 2^{m-1}n$, and this bound is tight if the cardinality of $\Sig$ is at least 2.
\label{thm:suffixtight}
 \end{theorem}
\begin{proof}
The upper bound follows from Proposition \ref{prop:suffix}.
To prove that the upper bound is tight, we show that all states in the direct product $\cB_m(b,a) \times \cD_n(a,b)$ are reachable and pairwise distinguishable.

\noin
{\bf Reachability.}
Let $B = \{0,1\}^{m-1}$ denote the state set of $\cB_m(b,a)$ and let $v_0$ denote the initial state of $\cB_m(b,a)$.
The initial state of the direct product is $(v_0,0)$.
Every state of the form $(v_0,q)$, where $0\le q \le n-1$, is reachable by $a^q$.
We observe that ${(\delta_T})_a = (0,1,\ldots,n-1)$ and ${(\delta_T})_b = (1,2,\ldots,n-1)$ (where $\delta_T$ is the transition function of $\cD_n(a,b)$) are both bijective on $\{ 0, 1, \ldots, n-1 \}$.
Therefore, by applying Lemma~\ref{lem:prodreach}, we see that every state in $B \times \{ 0, \ldots, n-1 \}$ is reachable.

\noin
{\bf Distinguishability.}
In this part of the proof, to simplify the notation, we simply write $w$ for the transformation induced by $w$ in the appropriate DFA. For example, if $u \in B$, then $uab^{n-2}$ is equivalent to $u\beta_a \beta_b^{n-2}$ or $u\beta_{ab^{n-2}}$.

First note the following facts about $\cB_m(b,a)$:
\bi
\item
The word $b^{m-1}$ sends 
all states
 to the final state $(1,1,\dotsc,1)$.
\item
The final state $(1,1,\dotsc,1)$ is fixed by all words in $\{a,b\}^*$.
\item
The letter $a$ permutes the states. Thus if $u$ and $v$ are distinct, then $ua$ and $va$ are distinct.
\item
Suppose $u = (u_1,u_2,\dotsc,u_{m-1})$ and $v = (v_1,v_2,\dotsc,v_{m-1})$ are states, and define $d(u,v)$ to be the largest integer $i$ such that $u_i \ne v_i$, or $0$ if the states are equal. 
If $d(u,v) = 1$, then $u$ and $v$ are distinguishable by $\eps$ (that is, one is final and one is non-final). 
\item
If $d(u,v) \ne 1$, then $b^{d(u,v)-1}$ sends $u$ to a state $u'$ and $v$ to a state $v'$ such that $d(u',v') = 1$. 
\ei

Now, let $(u,p)$ and $(v,q)$ be distinct states, where $u,v \in \{0,1\}^{m-1}$.

\noin
{\bf Case 1.} $p \ne q$. 
Without loss of generality, we can assume $u = v = (1,1,\dotsc,1)$; otherwise apply $b^{m-1}$. Choose a word $w$ that distinguishes $p$ and $q$ in $\cD_n(a,b)$; then $w$ distinguishes $(u,p)$ and $(v,q)$.

\noin
{\bf Case 2.} $p = q$ (and thus $u \ne v$).
We may assume without loss of generality that $u$ and $v$ differ in exactly one component, and that $p = n-1$.
Otherwise, first apply $b^{d(u,v)-1}$ to reach $(u',p')$ and $(v',p')$ such that $d(u',v') = 1$, and note that this implies $u'$ and $v'$ differ in exactly one component.
Then apply $a^{n-1-p'}$ to send $p'$ to $n-1$.

Suppose now that $u$ and $v$ differ in exactly one component and $p = n-1$.
Then $d(u,v)$ is the index of the component where $u$ and $v$ differ.
Furthermore, if we apply a word $w \in \{a,b\}^*$, then either $uw = vw$, or $uw$ and $vw$ differ in exactly one component and $d(uw,vw)$ is the index of this component.
So as long as $w$ does not  erase $u$ and $v$'s differing component, it can be used to shift the  differing component's index.

If $d(u,v) = 1$, then $(u,n-1)$ and $(v,n-1)$ are distinguishable by $\eps$.
So suppose $d(u,v) > 1$, and set $i = d(u,v)$.
Observe that:
\bi
\item
For all $k \ge 0$, we have $d(ua^k,va^k) \equiv i-k \pmod{m-1}$.
\item
For all $k \ge 0$, since $d(u,v) = i > 1$, we have $d(uba^k,vba^k) \equiv i-k-1 \pmod{m-1}$.
\ei
Since $a^n$ and $ba^{n-2}$ both fix $p = n-1$, it follows that:
\bi
\item
If we are in states $(u,n-1)$ and $(v,n-1)$ and apply $a^n$, we reach $(ua^n,n-1)$ and $(va^n,n-1)$ where $d(ua^n,va^n)$ is the unique element of $\{1,\dotsc,m-1\}$ equivalent to $i-n$ modulo $m-1$.
\item
If we are in states $(u,n-1)$ and $(v,n-1)$ and apply $ba^{n-2}$, we reach $(uba^{n-2},n-1)$ and $(vba^{n-2},n-1)$, where $d(uba^{n-2},vba^{n-2})$ is the unique element of $\{1,\dotsc,m-1\}$ equivalent to $i-(n-1)$ modulo $m-1$.
\ei
Let $x = a^n$ and $y = ba^{n-2}$. 
Apply $x^{i-1}$ to the states to reach $(ux^{i-1},n-1)$ and $(vx^{i-1},n-1)$, where $d(ux^{i-1},vx^{i-1})$ is the unique element of $\{1,\dotsc,m-1\}$ equivalent to $i-(i-1)n$ modulo $m-1$. 
We claim that we can now distinguish $(ux^{i-1},n-1)$ and $(vx^{i-1},n-1)$ by applying $y^k$ for some value $k \ge 0$.

We choose $k$ to be the least integer such that $d(ux^{i-1}y^k,vx^{i-1}y^k) = 1$.
Clearly if such a $k$ exists, then $y^k$ distinguishes the states, so we just have to show that $k$ exists.
Suppose for a contradiction that $k$ does not exist.
Observe then that
$d(ux^{i-1}y^\ell,vx^{i-1}y^\ell) > 1$ 
for all $\ell \ge 0$.
Otherwise, we can choose  a minimal $\ell$ so that 
$d(ux^{i-1}y^\ell,vx^{i-1}y^\ell) = 0$;
then we necessarily have $d(ux^{i-1}y^{\ell-1},vx^{i-1}y^{\ell-1}) = 1$, since the only way we can have $u'y = v'y$ is if $d(u',v') \le 1$. It follows then that we can take $k = \ell-1$.
Now, set $\ell = (i-1)(m-2)$.
Since 
$d(ux^{i-1}y^j,vx^{i-1}y^j) > 1$ 
for all $j \le \ell$, it follows that
$d(ux^{i-1}y^\ell,vx^{i-1}y^\ell)$ is the unique element of $\{1,\dotsc,m-1\}$ equivalent to $i-(i-1)n-\ell(n-1)$ modulo $m-1$. 
Indeed, each application of $y$ subtracts $n-1$ (modulo $m-1$) from the component where the bit tuples differ,
and since we always have $d(ux^{i-1}y^j,vx^{i-1}y^j) > 1$,
 the states are never 
mapped to the same state
 by the $b$ at the start of $y$.
But now, we have
{\small
\[ i-(i-1)n-\ell(n-1) = i - (i-1)n - (i-1)(m-2)(n-1) = i - (i-1)(n+(m-2)(n-1)). \]
}
Since $m-2 \equiv -1 \pmod{m-1}$, we have
\[ i-(i-1)n-\ell(n-1) \equiv i - (i-1)(n-n+1) \equiv i - (i-1) \equiv 1 \pmod{m-1}. \]
So in fact
$d(ux^{i-1}y^\ell,vx^{i-1}y^\ell) = 1$.
This is a contradiction, and 
so
the integer $k$ exists.
Thus if we set $w = x^{i-1}y^k$, the states $(u,n-1)$ and $(v,n-1)$ are distinguished by $w$ (note that both $x$ and $y$ fix the second component $n-1$). \qed
\end{proof}

\section{Factor Matching}
Let $T$ and $P$ be regular languages over an alphabet $\Sig$.
We want to find 
the worst-case state complexity of
the set $L$ of all the words of $T$ that have words of $P$ as factors. More formally, $L = \{ xwy \mid w \in P, xwy \in T \}=(\Sig^*P\Sig^*)\cap T$.
\begin{proposition}
\label{prop:factor}
 For $m,n\ge 3$, 
%if $P$ and $T$ are of state complexities $\le m$ and $\le n$, respectively, 
if $\kappa(P) \le m$ and $\kappa(T) \le n$,
 then $\kappa((\Sig^*P\Sig^*)\cap T) \le (2^{m-2}+1)n$.
  \end{proposition}
 \begin{proof}
The language $\Sig^*P\Sig^*$ is the two-sided ideal generated by $P$. It is known that the state complexity of $\Sig^*P\Sig^*$ is at most $2^{m-2}+1$~\cite{BJL13}.
Thus the complexity of the intersection with $T$ is at most  $(2^{m-2}+1)n$.
\qed
\end{proof}
\smallskip

To prove the bound is tight, we construct a witness that meets the bound.
Let $T_n(a,b)$ be accepted by the DFA $\mathcal{D}_n(a,b) = (Q_n, \Sigma, \delta_T, 0, \{ n-1 \})$, where $Q_n=\{0,1,\dots,n-1\}$, $\Sig=\{a,b\}$ and $\delta_T$ is defined by the transformations $a \colon (0, 1, \ldots, n-1)$ 
and
 $b \colon (1,2,\dots,n-2)$.
This DFA is minimal because the shortest word in $a^*$ accepted by state $q$ is $a^{n-1-q}$. 

\begin{figure}[ht!]
\unitlength 8.5pt
\begin{center}\begin{picture}(37,7)(0,3)
\gasset{Nh=1.8,Nw=3.8,Nmr=1.25,ELdist=0.4,loopdiam=1.5}
	{\small
\node(0)(1,7){0}\imark(0)
\node(1)(8,7){1}
\node(2)(15,7){2}
\node[Nframe=n](3dots)(22,7){$\dots$}
\node(n-2)(29,7){$m-2$}
\node(n-1)(36,7){$m-1$}\rmark(m-1)
\drawedge(0,1){$b$}
\drawedge(1,2){$a,b$}
\drawloop(0){$a$}
\drawloop(n-1){$a$}
\drawedge(2,3dots){$a,b$}
\drawedge(3dots,m-2){$a,b$}
\drawedge(m-2,m-1){$b$}
\drawedge[curvedepth= 4.0,ELdist=-1.0](m-1,0){$b$}
\drawedge[curvedepth= 2.0,ELdist=-1.0](m-2,1){$a$}
	}
\end{picture}\end{center}
\caption{Minimal DFA $\cD_m(b,a)$ of $P_m(b,a)$.}
\label{fig:Mfactor}
\end{figure}

%\newpage

It turns out 
that $\cD_n(a,b)$ and its dialect $\cD_m(b,a)$ act as  witnesses in the case of factor matching.
We denote the language of $\cD_m(b,a)$ by $P_m(b,a)$; the DFA $\cD_m(b,a)$ is shown in Figure~\ref{fig:Mfactor}.
Let $E = \Sig^{m-3}(a\Sig^{m-3})^*$.
Then the language accepted by the DFA of Figure~\ref{fig:Mfactor} is denoted by the regular expression 
$$
P_m(b,a)=(a \cup bE ba^*b)^* bEba^*.
$$
Now we have 
$$
\Sig^*P_m(b,a)\Sig^*=\Sig^*(a \cup bE ba^*b)^* bEba^*\Sig^* = 
\Sig^*bE b\Sig^*.$$
The new generator of the two-sided ideal is $G_m=bEb= b \Sig^{m-3}(a\Sig^{m-3})^* b$.

Before we prove
that the bound is tight,
we describe a DFA for the language $\Sig^*P_m(b,a)\Sig^*$. We will use binary $(m-2)$-tuples which we denote by $x =(x_1,\dots, x_{m-2})$.
%\newpage
\begin{definition}
Define the following DFA:
$$\cC_m(b,a)=( \{0,1\}^{m-2}\cup \{f\},\{a,b\},(0,\dots,0),\gamma,\{ f \} ),$$
where $\gamma(f, \sig) = f$ for all $\sig \in \Sig$, and
\[ 
\gamma((x_1,x_2,\dots,x_{m-3},x_{m-2}), \sigma) = 
\begin{cases}
(x_2,x_3,\dots, x_{m-2},x_1), & \text{if } \sigma = a;\\
(x_2,x_3,\dots, x_{m-2},1), & \text{if } \sigma = b \text {, } x_1=0;\\
f, & \text{if } \sigma = b \text {, } x_1=1.
\end{cases}
\]
In other words, if $x\ne f$, input $\sig=a$ shifts $x$ one position to the left cyclically; input $\sig = b$ shifts the
tuple to the left, losing the leftmost component and replacing $x_{m-2}$ by 1 if $x_1=0$.
Finally,
$\gamma$ sends the state to $f$ if $x_1=1$  and $\sig = b$, and all inputs are the identity on $f$.
DFA 
$\cC_5(b,a)$
 is shown in Figure~\ref{fig:C5}.
%\newpage

\begin{figure}[ht]
\unitlength 8.5pt
\begin{center}
  \begin{picture}(35,20)(-3,2)
    \gasset{Nh=1.8,Nw=4.5,Nmr=1.25,ELdist=0.4,loopdiam=1.5}
    {\footnotesize
    %\node(000)(0,15){$(0,0,0)$}\imark(000)
    \node(001)(8,12){$(0,0,1)$}
    \node(010)(16,17){$(0,1,0)$}
    \node(100)(24,20){$(1,0,0)$}
    \node(000)(0,12){$(0,0,0)$}\imark(000)
    \node(101)(24,14){$(1,0,1)$}
    \node(011)(16,7){$(0,1,1)$}
    \node(110)(24,10){$(1,1,0)$}
    \node(111)(24,4){$(1,1,1)$}
    \node(f)(32,12){$f$}\rmark(f)
    \drawloop(000){$a$}
    \drawloop(111){$a$}
\drawloop(f){$a,b$}
  \drawedge(000,001){$b$}
 \drawedge(001,010){$a$}
\drawedge(010,100){$a$}
\drawedge[curvedepth= 4,ELdist=.4](110,101){$a$}  
\drawedge[curvedepth= -5,ELdist=.4](100,001){$a$} 
\drawedge[curvedepth= -4,ELdist=.4](101,011){$a$} 
\drawedge(001,011){$b$}
\drawedge(010,101){$b$}
%    \drawedge(101,011){$a,b$}
\drawedge(011,111){$b$}
\drawedge(011,110){$a$}

\drawedge(100,f){$b$}
\drawedge(101,f){$b$}
\drawedge(110,f){$b$}
\drawedge(111,f){$b$}

%    \drawedge[ELdist=-2](110,101){$a,b$}
    }
  \end{picture}
\end{center}
\caption{The DFA $\cC_5
(b,a)
$ for $\Sig^*P_5(b,a)\Sig^*$.}
\label{fig:C5}
\end{figure}

\end{definition}

\begin{proposition}
\label{prop:reachfactor}
All the states of $\cC_m(b,a)$ are reachable and pairwise distinguishable.
\end{proposition}
\begin{proof}
Consider a state $(x_1,\dots,x_{m-2})$, and view it as the binary representation of a number $k$.
Then $k$ is reachable as in the proof of Proposition~\ref{prop:reachsuffix},
and $f$ is reached by applying 
$b$
to any state that has $x_1=1$.

We note that if a state $q = (x_1,\ldots,x_i,\ldots,x_{m-2})$ has $x_i = 1$, then $q$
accepts the word $a^{i-1}b$. 
Define $\mathbf{Ab}(q) = \{ a^{i-1}b \mid x_i = 1 \}$. 
As in Proposition~\ref{prop:reachsuffix}, each 
binary (that is, non-$f$) 
state has a unique binary representation, and so each of these
states has a unique $\mathbf{Ab}(q)$, which is a subset of all words accepted by $q$.
Therefore, if $p$ and $q$ are distinct binary states, they are pairwise distinguishable by words in $\mathbf{Ab}(p) \cup \mathbf{Ab}(q)$.
We observe that $f$ is the only final state, and is therefore distinguishable from every other state by $\epsilon$.
\qed
\end{proof}

In the example of Figure~\ref{fig:C5}, we have 
$\mathbf{Ab}(001)=\{aab\}$,
$\mathbf{Ab}(010)=\{ab\}$,
$\mathbf{Ab}(011)=\{ab,aab\}$,
$\mathbf{Ab}(100)=\{b\}$,
$\mathbf{Ab}(101)=\{b,aab\}$,
$\mathbf{Ab}(110)=\{b,ab\}$,
$\mathbf{Ab}(111)=\{b,ab,aab\}$.

Recall that the two-sided ideal $\Sig^* P_m(b,a)\Sig^*$ is generated by the language $G_m=b\Sig^{m-3}(a\Sig^{m-3})^*b$, that is, $\Sig^* P_m(b,a) \Sig^* = \Sig^* G_m \Sig^*$.
\begin{lemma}
$\cC_m(b,a)$ 
is isomorphic to the minimal DFA of $\Sig^*P_m(b,a)\Sig^*$.
\end{lemma}
\begin{proof}
First we prove that each state of $\cC_m(b,a)$ accepts $G_m\Sig^*$, and thus $\cC_m(b,a)$ accepts $\Sig^*G_m\Sig^* = \Sig^*P_m(b,a)\Sig^*$.
Since $f$ accepts $\Sig^*$, it also accepts
$G_m\Sig^*$.  In binary states of the form $(x_1, x_2, \ldots, x_{m-2})$, applying $b$ ``loads'' a 1 into $x_{m-2}$. 
Then
after applying 
a word from
$\Sig^{m-3}$, the resulting state will either be a binary state where $x_1 = 1$, or $f$.  If the current
state is $f$, then no matter what inputs are applied, the word will be accepted, and hence $G_m\Sig^*$ is accepted. If
the current state is a binary state with $x_1 = 1$, then applying $a$ will cycle the 1 at $x_1$ to $x_{m-2}$, and
applying 
a word from
$\Sig^{m-3}$
will either shift the 1 back to $x_1$ or 
move
to $f$ if another 1 in the state is
shifted to $x_1$ and $b$ is applied. Therefore, applying 
a word from
$(a\Sig^{m-3})^*$ 
from a state where $x_1 = 1$ will result in
either a binary state where $x_1 = 1$ or $f$, and applying 
a word from
$b\Sig^*$ 
from one of those states will result in $f$, so
$G_m\Sig^*$ is accepted. Therefore, $\Sig^*G_m\Sig^* \subseteq L(\cC_m)$.

We now show that $L(\cC_m) \subseteq \Sig^*G_m\Sig^*$. First, we observe that every word in $L(\cC_m)$ has a length of
at least $m-1$ and at least two $b$s: a $b$ to load a 1 into $x_{m-2}$, $m-3$ letters to shift the 1 to $x_1$, and a $b$
to 
move
to $f$. Let $w = \sig_1 \ldots \sig_k$ be a word in $L(\cC_m)$, where each $\sig_i \in \Sig$. Suppose that
the $j$-th letter of $w$ is what first causes a transition to $f$; in other words, $(0, \ldots, 0)\gamma_{\sig_1 \ldots
\sig_{j-1}} \ne f$ and $(0, \ldots, 0)\gamma_{\sig_1 \ldots \sig_j} = f$. The remaining letters in $w$, $\sig_{j+1}
\ldots \sig_k$, do not matter since they cannot cause a transition away from $f$, so we only need to consider the prefix
$w_j = \sig_1 \ldots \sig_j$. 

Now the rest of the argument is similar to the proof of Lemma \ref{lem:suffix-iso}.
Letter $\sig_j$ of $w_j$ must be a $b$.
Look at letter $\sig_{j-1-(m-3)}$. If this letter is a $b$, then $w_j$ 
is in
 $\Sig^*b\Sig^{m-3}b$, and so $w$ 
is in
 $\Sig^*b\Sig^{m-3}b\Sig^* \subseteq \Sig^* G_m \Sig^*$, and we are done.
If the letter $\sig_{j-1-(m-3)}$ is an $a$, we keep jumping back $m-2$ letters at a time until we find a $b$.
In other words, we choose $\ell \ge 0$ as small as possible such that $\sig_{j-1-(m-3)-\ell(m-2)} = b$.
If no such $\ell$ exists, then as in the proof of Lemma \ref{lem:suffix-iso}, one can show that $w_j$ must 
be in
$\Sig^i(a\Sig^{m-3})^*b$ with $i < m-2$ and that $w$ is not accepted.
So $\ell$ must exist, and therefore $w_j$ 
is in
 $\Sig^*b\Sig^{m-3}(a\Sig^{m-3})^\ell b$, which implies $w \in \Sig^*G_m\Sig^*$.
\qed
\end{proof}

We can now prove the following theorem:
\begin{theorem}
For $m,n\ge 3$, 
if $\kappa(P) \le m$ and  $\kappa(T) \le n$,
%if $P_m$ and $T_n$ are of state complexities $\le m$ and $\le n$, respectively, 
then $\kappa((\Sig^*P\Sig^*)\cap T) \le (2^{m-2}+1)n$, 
and this bound is tight if the cardinality of $\Sig$ is at least 2.
\label{thm:factortight}
\end{theorem}
\begin{proof}
The upper bound follows from Proposition~\ref{prop:factor}.
To prove that the upper bound is tight, we show that all states in the direct product $\cC_m(b,a) \times \cD_n(a,b)$ are reachable and pairwise distinguishable.

\noin
\textbf{Reachability.}
Let $C = \{0,1\}^{m-2} \cup \{f\}$ denote the state set of $\cC_m(b,a)$, and let $v_0$ denote the initial state of
$\cC_m(b,a)$.  The initial state of the direct product is $(v_0, 0)$.  Every state of the form $(v_0, q)$, where $0 \le
q \le n-1$, is reachable by $a^q$.  We observe that $(\delta_T)_a = (0, 1, \ldots, n-1)$ and $(\delta_T)_b = (1, 2,
\ldots, n-2)$ (where $\delta_T$ is the transition function of $\cD_n(a,b)$) are both bijective on $\{0, 1, \ldots,
n-1\}$.  Therefore, by applying Lemma~\ref{lem:prodreach}, we see that every state in $C \times \{0, \ldots, n-1\}$ is
reachable.

\noin
\textbf{Distinguishability.}
As before, to simplify the notation,
we write $w$ for the transformation induced by $w$ in the relevant DFA. %For example, if $u \in C$, then $uab^{n-2}$ is equivalent to $u\gamma_a \gamma_b^{n-2}$ or $u\gamma_{ab^{n-2}}$.

First, we note a few facts about $\cC_m(b,a)$:
\bi
\item The word $b^{m-1}$ sends every state to $f$.
\item Suppose $u$ and $v$ are states. Define the function $d(u,v)$ as follows:
\small{
  \[
    d(u,v) = \begin{cases}
      -1, & \text{if } u = v; \\
      0, & \text{if } u = f \text{ or } v = f; \\
      \min\{i \mid u_i \ne v_i\}, & \text{if } u = (u_1, \ldots, u_{m-2}) \text{ and } v = (v_1, \ldots, v_{m-2}). \\
    \end{cases}
  \]
}
\ei

Suppose we have two distinct states in $\cC_m(b,a) \times \cD_n(a,b)$: $(u, p)$ and $(v, q)$.

\noin
\textbf{Case 1.}
$p \ne q$. Assume that $u = v$; if not, apply $b^{m-1}$ to send both to $f$.
$(f, pb^{m-1})$ and $(f, qb^{m-1})$ can be distinguished by $a^{n-1-pb^{m-1}}$.

\noin
\textbf{Case 2.}
$p = q$ (so $u \ne v$). Assume that $d(u, v) = 0$; if not, apply $a^{d(u,v)-1}b$ to send either $u$ or $v$ to $f$.
Then we have the states $(ua^{d(u,v)-1}b, pa^{d(u,v)-1}b)$ and $(va^{d(u,v)-1}b, pa^{d(u,v)-1}b)$.
Let us define $p^\prime = pa^{d(u,v)-1}b$; the two states can be distinguished by $a^{n-1-p^\prime}$.
\qed
\end{proof}

\section{ Subsequence Matching }
Let $T$ and $P$ be regular languages over an alphabet $\Sig$. 
We are interested in finding 
the worst-case state complexity of
the set $L$ of all the words of $T$ that contain words in $P$ as subsequences.
The set of all words which contain words in $P$ as subsequences can be constructed using the shuffle operation, as $(\Sig^* \shu P)$.
Thus $L = (\Sig^* \shu P) \cap T$.

\begin{theorem}
For 
$m,n \ge 3$, 
if $\kappa(P) \le m$ and $\kappa(T) \le n$, then
$\kappa((\Sig^* \shu P) \cap T) \le (2^{m-2}+1)n$, and this bound is tight if $|\Sig| \ge m-1$.
\end{theorem}
\begin{proof}
Okhotin~\cite{Okh10} proved that if $\kappa(P) \le m$, then $(\Sig^* \shu P)$ has state complexity at most $2^{m-2}+1$, and this bound is tight. 
Okhotin's witness is the DFA $(Q_m,\Sig,\delta,0,\{m-1\})$, where $\Sig = \{a_1,\dotsc,a_{m-2}\}$ and $a_i \colon (i \to m-1)(0 \to i)$; the alphabet size $m-2$ cannot be reduced.

It follows that if $\kappa(P) \le m$ and $\kappa(T) \le n$, then the state complexity of $(\Sig^* \shu P) \cap T$ is at most $(2^{m-2}+1)n$.

 We define $P_m$ as a slight modification of Okhotin's witness, with $m-1$ letters instead of $m-2$.
Define $\cD_m = (Q_m,\Sig,\delta,0,\{m-1\})$ where $\Sig = \{a_1,\dotsc,a_{m-2},b\}$, $a_i \colon (i \to m-1)(0 \to i)$ as before, and $b \colon \id$. 
See Figure~\ref{fig:subsequence}.
Let $P_m$ be the language of $\cD_m$.

\begin{figure}[ht]
\unitlength 8.5pt
\begin{center}
  \begin{picture}(20,16)(0,-1)
    \gasset{Nh=1.8,Nw=3.5,Nmr=1.25,ELdist=0.4,loopdiam=1.5}
    \node(0)(0,6){0}\imark(0)
    \node(1)(10,12){1}
    \node(2)(10,6){2}
    \node(3)(10,0){3}
    \node(4)(20,6){4}\rmark(4)
    \drawedge(0,1){$a_1$}
    \drawedge(0,2){$a_2$}
    \drawedge(0,3){$a_3$}
    \drawedge(1,4){$a_1$}
    \drawedge(2,4){$a_2$}
    \drawedge(3,4){$a_3$}
    \drawloop(0){$b$}
    \drawloop(1){$a_2,a_3,b$}
    \drawloop(2){$a_1,a_3,b$}
    \drawloop(3){$a_1,a_2,b$}
    \drawloop(4){$a_1,a_2,a_3,b$}
  \end{picture}
\end{center}
\caption{DFA $\cD_5$ of $P_5$ for subsequence matching.}
\label{fig:subsequence}
\end{figure}

\begin{figure}[ht!]
\unitlength 8.5pt
\begin{center}\begin{picture}(37,8)(-0.5,4)
\gasset{Nh=1.8,Nw=3.5,Nmr=1.25,ELdist=0.4,loopdiam=1.5}
\node(0)(1,7){0}\imark(0)
\node(1)(8,7){1}
\node(2)(15,7){2}
\node[Nframe=n](3dots)(22,7){$\dots$}
\node(n-2)(29,7){$n-2$}
\node(n-1)(36,7){$n-1$}\rmark(n-1)
\drawloop(0){$\Sig \setminus \{b\}$}
\drawedge(0,1){$b$}
\drawloop(1){$\Sig \setminus \{b\}$}
\drawedge(1,2){$b$}
\drawloop(2){$\Sig \setminus \{b\}$}
\drawedge(2,3dots){$b$}
\drawedge(3dots,n-2){$b$}
\drawloop(n-2){$\Sig \setminus \{b\}$}
\drawedge(n-2,n-1){$b$}
\drawloop(n-1){$\Sig \setminus \{b\}$}
\drawedge[curvedepth=4, ELdist=-1.1](n-1,0){$b$}
\end{picture}\end{center}
\caption{DFA $\cA_n$ of $T_n$ for subsequence matching.}
\label{fig:subsequence2}
\end{figure}

For $T_n$ we use the language of the DFA $\cA_n = (Q_n,\Sig,\delta',0,\{n-1\})$ where $a_i \colon \id$ for $1 \le i \le m-2$ and $b \colon (0,1,\dotsc,n-1)$.
See Figure~\ref{fig:subsequence2}.

Let $\cS_m$ be a minimal DFA for the shuffle $(\Sig^* \shu P_m)$ with state set $S$ and initial state $s_0$.
Note that all states of $\cS_m$ are reachable from $s_0$, and pairwise distinguishable from each other, using words over $\{a_1,\dotsc,a_{m-2}\}$ (that is, without using $b$). This follows from the fact that our DFA $\cD_m$ for $P_m$ was constructed using Okhotin's witness as a base.

Consider the direct product $\cS_m \times \cA_n$, which recognizes $(\Sig^* \shu P_m) \cap T_n$. The states of the direct product have the form $(s,q)$ where $s$ is a state of $\cS_m$ and $q$ is a state of $\cA_n$.
The initial state of the direct product is $(s_0,0)$. By words over $\{a_1,\dotsc,a_{m-2}\}$
we can reach all states of the form $(s,0)$ for $s \in S$. Then by words over $b^*$ we reach all states $(s,q)$ for all $s \in S$ and $q \in Q$. So all $(2^{m-2}+1)n$ states of the direct product are reachable.

For distinguishability, consider two distinct states $(s,q)$ and $(s',q')$. The final state set of the direct product is $\{(s_F,n-1) \mid s_F \text{ is final in } \cS_m\}$. 
Suppose $q \ne q'$. 
Since $\cS_m$ is minimal, it has at most one empty state. Hence one of $s$ or $s'$ can be mapped to a final state by some word $w$ over $\{a_1,\dotsc,a_{m-2}\}$. If we have states $(s,q)$ and $(s',q')$ with one of $s$ or $s'$ final, then a word in $b^*$ distinguishes the states.

Now suppose $q = q'$; then we must have $s \ne s'$. 
Apply $b^{n-1-q}$ to reach states $(s,n-1)$ and $(s',n-1)$.
By minimality of $\cS_m$, there is a word over $\{a_1,\dotsc,a_{m-2}\}$ that distinguishes $s$ and $s'$; this word also distinguishes $(s,n-1)$ and $(s',n-1)$. \qed
\end{proof}

The following proposition shows that the alphabet size of our witness cannot be reduced: 
an alphabet of
$m-1$ letters is optimal for this operation.

\begin{proposition}
Let $P$ and $T$ be regular languages with $\kappa(P) \le m$ and $\kappa(T) \le n$, both over an alphabet $\Sig$ of size less than $m-1$.
If $n \ge 1$, then $\kappa((\Sig^* \shu P) \cap T) < (2^{m-2}+1)n$.
\end{proposition}
\begin{proof}
To prove this, we need to understand the structure of the minimal DFA 
for $\Sig^* \shu P$.
Okhotin~\cite{Okh10} proved that if 
$\cD = (Q,\Sig,\delta,q_0,F)$ recognizes $P$, then the NFA $\cN = (Q,\Sig,\Delta,q_0,F)$, where $\Delta(q,\sig) = \{q,q\delta_{\sig}\}$, recognizes $\Sig^* \shu P$. 
We can obtain 
a minimal DFA
$\cS$ 
for $\Sig^* \shu P$
by determinizing and minimizing this NFA.
It follows that we can view the states of $\cS$ as subsets of $Q$, and, if $S$ is a subset of $Q$,
 then
$S \Delta_\sig = S \cup S\delta_\sig$
 in $\cS$.
To simplify the notation, write $\sig'$ for 
$\Delta_{\sig}$ 
and $\sig$ for $\delta_\sig$; so the previous equation can be written as $S\sig' = S \cup S\sig$.

Consider the direct product of $\cS$ with an arbitrary $n$-state DFA. Assume without loss of generality that
 the DFA $\cD$ for $P$ has state set $Q_m$ and initial state $0$, and the arbitary $n$-state DFA 
has state set $Q_n$ and initial state $0$.
Then 
the initial state of the direct product is $(\{0\},0)$. 
The only way we can reach states of the form $(\{0\},q)$ with $q \ne 0$ is if $\{0\}\sig' = \{0\}$ for some letter $\sig$.
In other words, at least one letter must induce a self-loop on $\{0\}$, or else the maximal number of states in the direct product is not reachable.
Our alphabet has size strictly less than $m-1$, and one of the letters in our alphabet 
induces
 a self-loop on $\{0\}$, so there are at most $m-3$ letters that do not induce a self-loop on $\{0\}$.

Now, we mimic Okhotin's argument from Lemma 4.4 in~\cite{Okh10}. 
Notice that in the NFA $\cN$, we have $S \subseteq S\sig'$ for all $\sig \in \Sig$.
Thus every reachable subset of states in this NFA contains the initial state $0$.
Additionally, if two subsets $S$ and $T$ in the NFA $\cN$ both contain a final state, then they are indistinguishable in the DFA $\cS$, since from these sets we can only reach other sets containing a final state.
If $\cN$ has $k$ final states, then there are $2^{m-k-1}$ sets that contain $0$ but do not contain a final state, and the remaining sets are indistinguishable. It follows there are at most $2^{m-k-1}+1$ indistinguishability equivalence classes. 
If $k \ge 2$, this is strictly less than the upper bound. 
Thus we may assume that $k = 1$, that is, there is a unique accepting state.
To reach the upper bound, \e{all} sets which do not contain the accepting state must be reachable.

Consider 
subsets of states
in $\cN$ of the form $\{0,p\}$ for $p \ne 0$ and $p$ non-final; there are $m-2$ such sets, since there is
only one
 accepting state.
Since $S \subseteq S\sig'$ for all $\sig \in \Sig$, the only way we can reach a set $\{0,p\}$ is by a self-loop on $\{0,p\}$, or by a direct transition from a smaller set. But the only smaller reachable set is the initial set $\{0\}$. So if $\{0,p\}$ is reachable, then it is reachable by a direct transition from $\{0\}$. 

Now, we know one letter induces a self-loop on $\{0\}$, so it is not useful for reaching states of the form $\{0,p\}$.
We have at most $m-3$ letters that do not induce a self-loop on $\{0\}$, so we can reach at most $m-3$ sets of the form $\{0,p\}$. Since there are $m-2$ such sets, at least one set must be unreachable, and thus the upper bound on the state complexity of $(\Sig^* \shu P) \cap T$ cannot be reached. \qed
\end{proof}

\section{Matching a Pattern Consisting of a Single Word}

We now consider the case where 
the pattern
$P$ consists of a single nonempty word $w$. 
Note that if the state complexity of 
$P = \{w\}$ 
is $m$, then $w$ is of length $m-2$.

Throughout this entire section, we fix $w = a_1 \dotsb a_{m-2}$, where $a_i \in \Sig$ for $1 \le i \le m-2$.
Let $w_0 = \eps$ and for $1 \le i \le m-2$, let $w_i = a_1 \dotsb a_i$.
We write $W = \{w_0,w_1,\dotsc,w_{m-2}\}$ for the set of all prefixes of $w$.

\subsection{Matching a Single Prefix}
\label{sec:prefixword}
\begin{theorem}
Suppose $m  \ge 3$ and $n \ge 2$.
If $w$ is a non-empty word, $\kappa(\{w\}) \le m$ and $\kappa(T) \le n$
then we have
\[ \kappa(w\Sig^* \cap T) \le \begin{cases}
m+n-1, & \text{if $|\Sig| \ge 2$;} \\
m+n-2, & \text{if $|\Sig| = 1$.}
\end{cases} \]
Furthermore, these upper bounds are tight.
\end{theorem}

\begin{remark}
\label{rmk:unary}
When $|\Sig| = 1$ (that is, $P$ and $T$ are languages over a unary alphabet), the tight upper bound $m+n-2$ actually holds in \e{all eight cases} we consider in this paper. This is because if $L$ is a language over  a unary alphabet $\Sig$, then the ideals $L\Sig^*$, $\Sig^*L$, $\Sig^*L\Sig^*$ and $\Sig^* \shu L$ coincide; thus the prefix, suffix, factor and subsequence matching cases coincide. Furthermore, if $\Sig = \{a\}$ and $L$ is non-empty, then we have $L\Sig^* = a^i\Sig^*$, where $a^i$ is the shortest word in $L$. Thus the single-word and multi-word cases coincide as well.
\end{remark}

\begin{proof}
We first derive upper bounds for the two cases of $|\Sig|$.

\noin {\bf Upper Bounds:}
Let  $\mathcal{D}_T = (Q, \Sigma, \delta, q_0, F_{ T})$, where
$Q=\{q_0,\dots,q_{n-1}\}$, be a  DFA accepting $T$.
Let $P = \{w\}$ and let the minimal DFA of $P$ be $\mathcal{D}_P = (W \cup \{\emp\}, \Sigma, \alpha, w_0, \{w_{m-2}\})$.
Here $w_{m-2}$ is the only final state, and $\emp$ is the empty state. 
Define $\alpha$ as follows: for $0 \le i \le m-2$, we set
\[ \alpha(w_i,  a) =
\begin{cases}
  w_{i+1}, & \text{if }  a=a_i ; \\
  \emp, & \text{otherwise. } 
\end{cases}
\]
Also define $\alpha(\emp, a) = \emp$ for all $ a \in \Sig$.
Let the state reached by $w$ in $\cD_T$ be $q_r=\delta(q_0,w)$; we construct a DFA $\cD_{ L}$ that accepts $ L = (w\Sig^*)\cap T$.
As shown in Figure~\ref{fig:word}, let 
$\mathcal{D}_L = (Q \cup (W \setminus \{ w_{m-2} \}) \cup \{\emp\}, \Sigma, \beta, w_0,  F_T)$, where  $\beta$ is defined as follows: for $q\in Q \cup (W \setminus \{ w_{m-2} \}) \cup \{\emp\}$ and $ a\in\Sig$,

\[ \beta(q,  a) =
\begin{cases}
  \delta(q,  a), & \text{if } q \in Q; \\
  \alpha(q,  a), & \text{if } q \in W \setminus \{ 
w_{m-2},
w_{m-3}\}; \\
  q_r, & \text{if } q = w_{m-3}, \text{ and }  a=a_{m-2}; \\
  \emp, & \text{otherwise. } 
\end{cases}
\]
%
%\newpage
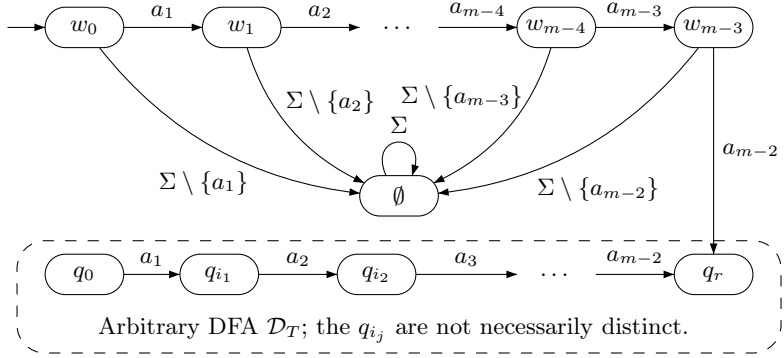
\begin{figure}[ht]
\unitlength 8.5pt
\begin{center}
  \begin{picture}(28,16)(0,-4)
  
\gasset{Nh=5.0,Nw=34,Nmr=1.25,ELdist=0.4,loopdiam=1.5}
\node[dash={.5}0](box1)(14,-1){}
\node[Nframe=n](name)(14,-2.5){\small Arbitrary DFA $\cD_{ T}$; the $q_{i_j}$ are not necessarily distinct.
}

    \gasset{Nh=1.8,Nw=3.5,Nmr=1.25,ELdist=0.4,loopdiam=1.5}
    {\small
      \node(p0)(0,11){$w_0$}\imark(p0)
      \node(p1)(7,11){$w_1$}
    \node[Nframe=n](p3dots)(14,11){$\dots$}
      \node(m-4)(21,11){$w_{m-4}$}
      \node(m-3)(28,11){$w_{m-3}$}
      \node(m-1)(14,3.5){$\emp$}
       \node(0)(0,0){$q_0$}
       \node(1)(6,0){$q_{i_1}$}
        \node(2)(13,0){$q_{i_2}$}
 %   \node[Nframe=n](3dots)(12,3){$\dots$}
%      \node(r-1)(18,3){$q_{r-1}$}\rmark(r-1)
    \node(r)(28,0){$q_r$}
    \node[Nframe=n](3dots2)(21,0){$\dots$}
%    {\scriptsize
%      \node(n-2)(36,7){$q_{n-2}$}
%      \node(n-1)(42,7){$q_{n-1}$}
 %   }
    \drawedge(p0,p1){$a_1$}
    \drawedge(p1,p3dots){$a_2$}
    \drawedge(p3dots,m-4){$a_{m-4}$}
    \drawedge(m-4,m-3){$a_{m-3}$}
    \drawedge(m-3,r){$a_{m-2}$}
      \drawedge[curvedepth=- 2,ELdist=-3](p0,m-1){$\Sig\setminus \{a_1\}$}
     \drawedge[curvedepth= 2,ELdist=.2](m-3,m-1){$\Sig\setminus \{ a_{m-2} \}$}
   \drawedge[curvedepth= -1.5,ELdist=0.3](p1,m-1){$\Sig\setminus\{a_2\}$}
 \drawedge[curvedepth= 1.5,ELdist=-4.8](m-4,m-1){$\Sig\setminus\{a_{m-3}\}$}
 \drawedge(1,2){$a_{ 2}$}
%    \drawedge(3dots,r-1){$a$}
%    \drawedge(r-1,r){$a$}
   \drawedge(0,1){$a_1$}
 \drawedge(3dots2,r){$a_{m-2}$}
  \drawedge(2,3dots2){$a_3$}
%    \drawedge(n-2,n-1){$a$}
%    {\gasset{curvedepth=3}
 %     \drawedge(n-1,0){$a$}
 \drawloop(m-1){$\Sig$}
  }
  \end{picture}
\end{center}
\caption{DFA 
$\cD_L$
for matching a single prefix.
The final state set $F_T$ is a subset of the states from the arbitrary DFA $\cD_T$; final states are not marked on the diagram.
}
\label{fig:word}
\end{figure}
Recall that in a DFA $\cD$, if state $q$ is reached from the initial state by a word $u$, then the language of $q$ is equal to the quotient of $L(\cD)$ by $u$.
Thus the language of state $q_r$ is the quotient of $T$ by $w$,
that is, the set $w^{-1}T = \{ y \in \Sig^* \mid wy \in T\}$.
The DFA $\cD_L$ accepts a word $x$ if and only if it has the form $wy$ for $y \in w^{-1}T$; we need the prefix $w$ to reach the arbitrary DFA $\cD_T$, and $w$ must be followed by a word that sends $q_r$ to an accepting state, that is, a word $y$ in the language $w^{-1}T$ of $q_r$.
So $L = \{ wy \mid y \in w^{-1}T \} = \{ wy \mid y \in \Sig^*, wy \in T\} = w\Sig^* \cap T$.
That is,
$ L$ is the set of all words of $T$ that begin with $w$, as required. It follows that the state complexity of $ L$ is less than or equal to $m+n-1$. If $|\Sig|=1$, all the $\Sig\setminus \{a_i\}$ are empty and state $\emp$ is not needed. Hence the state complexity of $ L$ is less than or equal to $m+n-2$ in this case.

\noin
{\bf Lower Bound, $|\Sig|=1$:} 
Let $m\ge 3$ and $P_m(a)=\{a^{m-2}\}$.
Let 
$n \ge 2$, 
and let $T_n(a)$ be the language of the DFA 
$\cD_n(a)=(Q_n, \{a\}, \delta_1,0,\{{r-1}\})$, where 
$\delta_1$ is defined by $a\colon (0,1,\dots,n-1)$, and
$r = \delta_1(0, a^{m-2})$.
Let $\cD_{ L}$ be the DFA shown in Figure~\ref{fig:L1}
for the language $ L = P_m(a)\Sig^* \cap T_n(a)$.
Obviously $\cD_{ L}$ has $m+n-2$ states and they are all  reachable. Since the shortest word accepted from any state is distinct from that of any other state, all the states are pairwise distinguishable.
Hence $P_m(a)$ and $T_n(a)$ constitute witnesses that meet the required bound. 
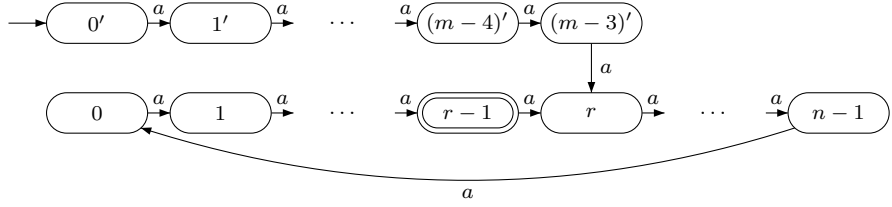
\begin{figure}[ht]
\unitlength 8.5pt
\begin{center}
  \begin{picture}(37,10)(-2.3,3.5)
    \gasset{Nh=1.8,Nw=4.5,Nmr=1.25,ELdist=0.4,loopdiam=1.5}
    {\footnotesize
      \node(p0)(0,11){$0'$}\imark(p0)
      \node(p1)(5.5,11){$1'$}
    \node[Nframe=n](p3dots)(11,11){$\dots$}
      \node(m-4)(16.5,11){${(m-4)'}$}
      \node(m-3)(22,11){${(m-3)'}$}
      \node(0)(0,7){$0$}
      \node(1)(5.5,7){$1$}
    \node[Nframe=n](3dots)(11,7){$\dots$}
        \node(r-1)(16.5,7){${r-1}$}\rmark(r-1)
      \node(r)(22,7){$r$}
    \node[Nframe=n](3dots2)(27.5,7){$\dots$}
     % \node(n-2)(33,7){${n-2}$}
      \node(n-1)(33,7){${n-1}$}
      \drawedge(p0,p1){$a$}
    \drawedge(p1,p3dots){$a$}
    \drawedge(p3dots,m-4){$a$}
    \drawedge(m-4,m-3){$a$}
    \drawedge(m-3,r){$a$}
    \drawedge(0,1){$a$}
    \drawedge(1,3dots){$a$}
    \drawedge(3dots,r-1){$a$}
    \drawedge(r-1,r){$a$}
    \drawedge(r,3dots2){$a$}
    \drawedge(3dots2,n-1){$a$}
    %\drawedge(n-2,n-1){$a$}
    {\gasset{curvedepth=3}
      \drawedge(n-1,0){$a$}
   }
    }
  \end{picture}
\end{center}
\caption{Minimal DFA of $ L$ for the case $|\Sig|=1$.}\label{fig:L1}
\end{figure}

\begin{figure}[ht]
\unitlength 8.5pt
\begin{center}
  \begin{picture}(37,18)(-2.3,1)
    \gasset{Nh=1.8,Nw=4.5,Nmr=1.25,ELdist=0.4,loopdiam=1.5}
    {\footnotesize
    \node(m-1)(11,15){$\emp$}
      \node(p0)(0,11){$0'$}\imark(p0)
      \node(p1)(5.5,11){$1'$}
    \node[Nframe=n](p3dots)(11,11){$\dots$}
      \node(m-4)(16.5,11){${(m-4)'}$}
      \node(m-3)(22,11){${(m-3)'}$}
      \node(0)(0,5){$0$}
      \node(1)(5.5,5){$1$}
    \node[Nframe=n](3dots)(11,5){$\dots$}
        \node(r-1)(16.5,5){${r-1}$}\rmark(r-1)
      \node(r)(22,5){$r$}
    \node[Nframe=n](3dots2)(27.5,5){$\dots$}
     % \node(n-2)(33,7){${n-2}$}
      \node(n-1)(33,5){${n-1}$}
     \drawedge(p0,p1){$a$}
     \drawedge[curvedepth= 1.5,ELdist=.6](p0,m-1){$b$}
    \drawedge(p1,p3dots){$a$}
    \drawedge(p1,m-1){$b$}
    \drawedge(p3dots,m-4){$a$}
    \drawedge(m-4,m-1){$b$}
    \drawedge(m-4,m-3){$a$}
    \drawedge(m-3,r){$a$}
    \drawedge[curvedepth= -1.5,ELdist=-1.4](m-3,m-1){$b$}
    \drawedge(0,1){$a$}
    \drawedge(1,3dots){$a$}
    \drawedge(3dots,r-1){$a$}
    \drawedge(r-1,r){$a$}
    \drawedge(r,3dots2){$a$}
    \drawedge(3dots2,n-1){$a$}
    %\drawedge(n-2,n-1){$a$}
      \drawedge[curvedepth= 4.5,ELdist=.6](n-1,0){$a$}
      \drawloop(0){$b$}
      \drawloop(1){$b$}
      \drawloop(r-1){$b$}
       \drawloop(n-1){$b$}
      \drawloop(m-1){$b$}
      \drawloop[loopangle=-90](r){$b$}
    }
  \end{picture}
\end{center}
\caption{Minimal DFA of $ L$ for the 
prefix case with
$|\Sig|>1$.
}
\label{fig:L2}
\end{figure}
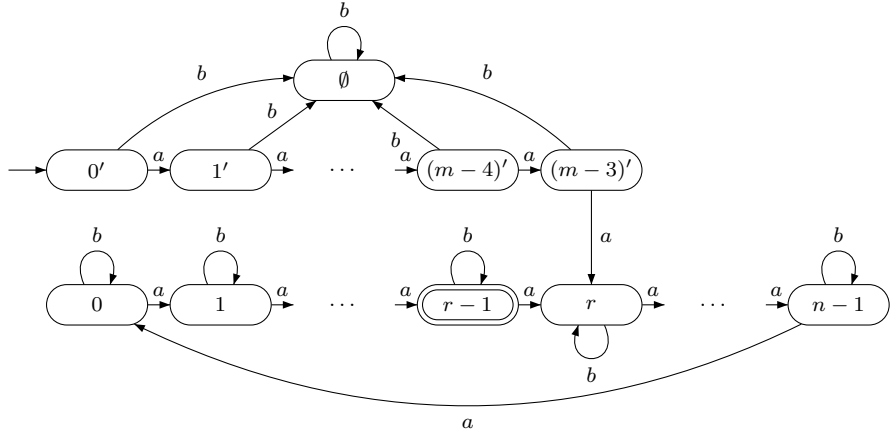

\noin{\bf Lower Bound, $|\Sig|\ge 2$:}
Let 
$m\ge3$ 
and $P_m(a,b)=\{a^{m-2}\}$. 
Let 
$n\ge2$ 
and let
$T_n(a,b)$ be the language of the DFA 
$\cD_n(a,b)=(Q_n,\{a,b\},\delta_2,0,\{r-1\})$ where
$\delta_2$ is defined by $a \colon (0, 1, \ldots, n-1)$ and $b \colon \id$, and $r = \delta_2(0, a^{m-2})$.
Construct the DFA $\cD_{ L}$ 
for the language ${ L} = P_m(a,b)\Sig^* \cap T_n(a,b)$
as is shown in Figure~\ref{fig:L2}.
It is clear that all the states are reachable and distinguishable by their shortest accepted words.
 \qed
 \end{proof}

\subsection{Matching a Single Suffix }

Let $w,x,y,z \in \Sig^*$. We introduce some notation:
\bi
\item
$x \prec_p y$ means $x$ is a proper prefix of $y$, and $x \preceq_p y $ means $x$ is a prefix of $y$.
\item
$x \succ_s y$ means $x$ has $y$ as a proper suffix, and $x \succeq_s y$ means $x$ has $y$ as a suffix.
\item
If $x \succeq_s y$ and $y \preceq_p z$, we say $y$ is a \e{bridge} from $x$ to $z$ or that $y$ \e{connects} $x$ to $z$.
We also denote this by $x \to y \to z$.
\item
$x \sto y \sto z$ means that $y$ is the \e{longest} bridge from $x$ to $z$.
That is, $x \to y \to z$, and whenever $x \to w \to z$ we have $|w| \le |y|$.
Equivalently, $y$ is the longest suffix of $x$ that is also a prefix of $z$.
\ei
We will readily use the following properties of these relations:
\bi
\item
For $a \in \Sig$, we have $x \preceq_p y \iff ax \preceq_p ay$.
\item
For $a \in \Sig$, we have 
$x \succeq_s y \iff xa \succeq_s ya$.
\item
If $x \ne \eps$ and $y$ starts with $a \in \Sig$ and $x \preceq_p y$, then $x$ starts with $a$.
\item
If $y \ne \eps$ and $x$ ends with $a \in \Sig$ and  
$x \succeq_s y$, 
then $y$ ends with $a$.
\item
If $x \preceq_p z$ and $y \preceq_p z$ and $|x| \le |y|$, then $x \preceq_p y$.
\item
If $z \succeq_s x$ and $z \succeq_s y$ and $|x| \ge |y|$, then $x \succeq_s y$.
\ei

\begin{proposition}
\label{prop:suffixword}
If the state complexity of $\{w\}$ is $m$, then the state complexity of $\Sig^*w$ is $m-1$. 
\end{proposition}
\begin{proof}

Let $\cA =( W, \Sig, \delta_\cA, w_0, \{w_{m-2}\} )$ be the DFA with transitions defined as follows: for all $a \in \Sig$ and $w_i \in W$, we have $w_ia \sto \delta_\cA(w_i,a) \sto w$.
That is, $\delta_\cA(w_i,a)$ is defined to be the maximal-length bridge from $w_ia$ to $w$, or equivalently, the longest suffix of $w_ia$ that is also a prefix of $w$.
Note that if $a = a_{i+1}$, then $\delta_\cA(w_i,a) = w_{i+1}$.

 We observe that every state $w_i\in W$ is reachable from $w_0$ by the word $w_i$, and that
each
 state $w_i$ is distinguished from
all
other states 
by $a_{i+1}\cdots a_{m-2}$.
It remains to be shown that 
$\Sig^*w = L(\cA)$.
In the following, for convenience, we simply write $\delta$ rather than $\delta_\cA$.

We claim that for $x \in \Sig^*$, we have $w_ix \sto \delta(w_i,x) \sto w$. That is, the defining property of the transition function extends nicely to words. Recall that the extension of $\delta$ to words is defined inductively in terms of the behavior of $\delta$ on letters, so it is not immediately clear that this property carries over to words.

We prove this claim by induction on $|x|$. 
If $x = \eps$, this is clear. 
Now suppose $x = ya$ for some $y \in \Sig^*$ and $a \in \Sig$, and that $w_iy \sto \delta(w_i,y) \sto w$. 
Let $\delta(w_i,y) = w_j$ and let $\delta(w_i,x)= \delta(w_j, a) = w_k$. 
We want to show that $w_ix \sto w_k \sto w$.

First we show that $w_ix \to w_k \to w$.
We know $w_k \preceq_p w$, so it remains to show that $w_ix \succeq_s w_k$.
Since $w_k = \delta(w_i,x) = \delta(w_j,a)$, by definition we have $w_ja \sto w_k \sto w$. 
Since $\delta(w_i,y) = w_j$, we have $w_iy \sto w_j \sto w$.
In particular, $w_iy \succeq_s w_j$ and thus $w_ix = w_iya \succeq_s w_ja$.
Thus $w_ix \succeq_s w_ja \succeq_s w_k$ as required.

Next, we show that whenever $w_ix \to w_\ell \to w$, we have $|w_\ell| \le |w_k|$.
If $w_\ell = \eps$, this is immediate, so suppose $w_\ell \ne \eps$.
Since $w_ix = w_iya \succeq_s w_\ell$, and $w_\ell$ is non-empty, it follow that $w_\ell$ ends with $a$.
Thus $w_\ell = w_{\ell-1}a$.
Since $w_iya \succeq_s w_{\ell-1}a$, we have $w_iy \succeq_s w_{\ell-1}$.
Additionally, $w_{\ell-1} \preceq_p w$, so $w_iy \to w_{\ell-1} \to w$.
Since $w_iy \sto w_j \sto w$, we have $|w_{\ell-1}| \le |w_j|$.
Since $w_iy \succeq_s w_j$ and $w_iy \succeq_s w_{\ell-1}$ and $|w_j| \ge |w_{\ell-1}|$, we have $w_j \succeq_s w_{\ell-1}$.
Thus $w_ja \succeq_s w_{\ell-1}a = w_\ell$.
It follows that $w_ja \to w_\ell \to w$.
But
recall that $\delta(w_i, x) = \delta(w_j, a) = w_k$, so
$w_ja \sto w_k \sto w$,
and
$|w_\ell| \le |w_k|$ as required.

Now, we show that $\cA$ accepts the language $\Sig^*w$.
Suppose $x \in \Sig^*w$ and write $x = yw$. The initial state of $\cA$ is $w_0 = \eps$.
We have $yw \sto \delta(\eps,yw) \sto w$, that is, $\delta(\eps,yw)$ is the longest suffix of $yw$ that is also a prefix of $w$.
But this longest suffix is simply $w$ itself, which is the final state. So $x$ is accepted.
Conversely, suppose $x \in \Sig^*$ is accepted by $\cA$. Then $\delta(\eps,x) = w$, and thus $x \sto w \sto w$ by definition. In particular, this means $x \succeq_s w$, and so $x \in \Sig^*w$.
\qed
\end{proof}

Our next goal is to establish an upper bound on the state complexity of $\Sig^*w \cap T$.
The upper bound in this case is quite complicated to derive.
Suppose $w$ has state complexity $m$ and $T$ has state complexity 
at most
 $n$, 
for $m \ge 3$ and $n \ge 2$.
Let $\cA$ be the 
$(m-1)$-state
DFA for $\Sig^*w$ defined in Proposition~\ref{prop:suffixword}, and let $\cD$ be an $n$-state DFA for $T$ with state set $Q_n$, transition function $\alpha$, and final state set $F$.
The direct product $\cA \times \cD$ with final state set $\{w\} \times F$ recognizes $\Sig^*w \cap T$.
We claim that this direct product has at most $(m-1)n-(m-2)$ reachable and pairwise distinguishable states,
and thus the state complexity of $\Sig^*w \cap T$ is at most $(m-1)n-(m-2)$.

Since $\cA$ has $m-1$ states and $\cD$ has $n$ states, there are at most $(m-1)n$ reachable states.
It will suffice show that for each word $w_i$ with $1 \le i \le m-2$, there exists a word $w_{f(i)} \ne w_i$ and a state $p_i \in Q_n$ such that $(w_i,p_i)$ is  indistinguishable from $(w_{f(i)},p_i)$.
This gives $m-2$ states that are each indistinguishable from another state, establishing the upper bound.

We choose $f(i)$ so that $w_i \sto w_{f(i)} \sto w_{i-1}$. In other words, $w_{f(i)}$ is the longest suffix of $w_i$ that is also a \e{proper} prefix of $w_i$.
To find $p_i$, first observe that there exists a non-final state $q \in Q_n$ and a state $r \in Q_n$ such that $\alpha(r,w) = q$. Indeed, if no such states existed, then for all states $r$, the state $\alpha(r,w)$ would be final. Thus we would have $\Sig^*w \subseteq T$, and the complexity of $\Sig^*w \cap T = \Sig^*w$ would be $m-1$, which is lower than our upper bound
since we are assuming $n \ge 2$.
Now, set $p_i = \alpha(r,w_i)$, and note that $\alpha(p_i,a_{i+1}) = p_{i+1}$, and $\alpha(p_i,a_{i+1} \dotsb a_{m-2}) = q$.

\begin{lemma}
\label{lem:suffixword_equal}
If $i < m-2$ and $a \ne a_{i+1}$, or if $i = m-2$, then 
$\delta_\cA(w_i,a) = \delta_\cA(w_{f(i)},a)$.
\end{lemma}
\begin{proof}
Let $w_j = \delta_\cA(w_i,a)$, so that $w_ia \sto w_j \sto w$.
Let $w_k = \delta_\cA(w_{f(i)},a)$, so that $w_{f(i)}a \sto w_k \sto w$.
We claim $j = k$.
To see that $j \ge k$, note that $w_i \succeq_s w_{f(i)}$, so $w_ia \succeq_s w_{f(i)}a \succeq_s w_k$.
Thus $w_ia \to w_k \to w$, but $w_ia \sto w_j \sto w$, which implies $|w_k| \le |w_j|$ and so $j \ge k$.
To see that $j \le k$, we consider  six cases:
\bi
\item
$w_j = \eps$. Then $j = 0$, so clearly $j \le k$.
\item
$w_j = a$. 
Then $w_{f(i)}a \to w_j \to w$.
Since $w_{f(i)}a \sto w_k \sto w$, we have $|w_j| \le |w_k|$ and thus $j \le k$.
\item
$f(i) = 0$ and $|w_j| \ge 2$. 
Since $|w_j| \ge 2$, we can write $w_j = w_{j-1}a_{j}$ with $w_{j-1}$ non-empty. 
Since $w_ia \succeq_s w_{j-1}a_j$, we have $w_i \succeq_s w_{j-1}$.
Now, note that $w_j = \delta_\cA(w_i,a)$ has length at most $i+1$, and this length is attained if and only if 
$i < m-2$ and
$a = a_{i+1}$.
We are assuming 
that either
$a \ne a_{i+1}$
or $i = m-2$;
in either case $|w_j| \le i$.
This means $j-1 \le i-1$ and it follows that $w_{j-1} \preceq_p w_{i-1}$.
Thus $w_i \to w_{j-1} \to w_{i-1}$.
Since $w_i \sto w_{f(i)} \sto w_{i-1}$, it follows that $j-1 \le f(i) = 0$, implying $j \le 1$.
This contradicts the assumption that $|w_j| \ge 2$, so this case cannot occur.
\item
$f(i) > 0$ and $2 \le |w_j| \le f(i)+1$. 
Since $w_i \sto w_{f(i)} \sto w_{i-1}$, we have $w_i \succeq_s w_{f(i)}$, and thus $w_ia \succeq_s w_{f(i)}a$.
Also, since $w_ia \sto w_j \sto w$ we have  $w_ia \succeq_s w_j$.
 Since $|w_{f(i)}a| = f(i)+1 \ge |w_j|$, it follows that $w_{f(i)}a \succeq_s w_j$.
Then $w_{f(i)}a \to w_j \to w$, but we have $w_{f(i)}a \sto w_k \sto w$, so $|w_j| \le |w_k|$ and thus $j \le k$.
\item
$f(i) > 0$ and $f(i)+1 < |w_j| < i+1$. 
Since $w_ia \succeq_s w_j$ and $w_j$ is non-empty, we can write $w_j = w_{j-1}a$.
Then $w_i \succeq_s w_{j-1}$.
Also, since $j < i+1$ we have $j-1 < i$, and so $w_{j-1} \preceq_p w_{i-1}$.
It follows that $w_i \to w_{j-1} \to w_{i-1}$.
Since $w_i \sto w_{f(i)} \sto w_{i-1}$ we have $j-1 \le f(i)$, and thus $j \le f(i)+1$. 
This contradicts the assumption that $j > f(i)+1$, so this case cannot occur.
\item
$|w_j| \ge i+1$. 
If $i = m-2$, this is impossible.
If $i < m-2$, this can only occur if $a = a_{i+1}$, but we are assuming $a \ne a_{i+1}$. So this case cannot occur.
\ei
This shows that $j = k$, and thus $w_j = w_k$. 
That is, $\delta_\cA(w_i,a) = \delta_\cA(w_{f(i)},a)$. \qed
\end{proof}

\begin{lemma}
\label{lem:suffixword_next}
If $i < m-2$, then $\delta_\cA(w_{f(i)},a_{i+1}) = w_{f(i+1)}$.
\end{lemma}
\begin{proof}
First we prove the following fact: $f(i+1) \le f(i)+1$.
If $f(i+1) = 0$, this is immediate, so assume $f(i+1) > 0$.
Since $f(i+1) > 0$, the word $w_{f(i+1)}$ is non-empty and thus $w_{f(i+1)}$ ends with $a_{i+1}$.
We can write $w_{f(i+1)} = w_{f(i+1)-1}a_{i+1}$.
Since $w_{i+1} \sto w_{f(i+1)} = w_{f(i+1)-1}a_{i+1} \sto w_i$,
in particular we have 
$w_{i+1} = w_ia_{i+1} \succeq_s w_{f(i+1)-1}a_{i+1}$, and so $w_i \succeq_s w_{f(i+1)-1}$.
Also, since $w_{f(i+1)-1}a_{i+1} \preceq_p w_i$, we have $w_{f(i+1)-1} \preceq_p w_{i-1}$.
It follows that $w_i \to w_{f(i+1)-1} \to w_{i-1}$.
Since $w_i \sto w_{f(i)} \sto w_{i-1}$, we have $f(i+1)-1 \le f(i)$.
Thus $f(i+1) \le f(i)+1$ as required.

Now, let $\delta_\cA(w_{f(i)},a_{i+1}) = w_j$.
Then $w_{f(i)}a_{i+1} \sto w_j \sto w$.
We have $w_i \sto w_{f(i)} \sto w_{i-1}$, and thus $w_i \succeq_s w_{f(i)}$.
Thus $w_ia_{i+1} = w_{i+1} \succeq_s w_{f(i)}a_{i+1} \succeq_s w_j$.
Also, since $f(i) < i$ and $j \le f(i)+1$, we have $j \le i$.
This implies $w_j \preceq_p w_i$.
It follows that $w_{i+1} \to w_j \to w_i$.
Since $w_{i+1} \sto w_{f(i+1)} \sto w_i$, we have $|w_j| \le |w_{f(i+1)}|$.

We noted above that $w_{i+1} \succeq_s w_{f(i)}a_{i+1}$, and we also have $w_{i+1} \succeq_s w_{f(i+1)}$.
Since $|w_{f(i)}a_{i+1}| = f(i)+1 \ge f(i+1) = |w_{f(i+1)}|$, it follows that $w_{f(i)}a_{i+1} \succeq_s w_{f(i+1)}$.
Hence $w_{f(i)}a_{i+1} \to w_{f(i+1)} \to w$.
Since $w_{f(i)}a_{i+1} \sto w_j \sto w$, we have $|w_{f(i+1)}| \le |w_j|$.
So $|w_j| = |w_{f(i+1)}|$, but both words are prefixes of $w$, so in fact $w_j = w_{f(i+1)}$ as required. \qed
\end{proof}

We can now establish the upper bound.
\begin{proposition}
Suppose $m \ge 3$ and $n \ge 2$.
If $w$ is non-empty, $\kappa(\{w\}) \le m$, and $\kappa(T) \le n$, then we have $\kappa(\Sig^*w \cap T) \le (m-1)n-(m-2)$.
\end{proposition}
\begin{proof}
It suffices to prove that 
states $(w_i,p_i)$ and $(w_{f(i)},p_i)$ are indistinguishable for $1 \le i \le m-2$.
We proceed by induction on the value $m-2-i$.

The base case is $m-2-i = 0$, that is, $i =m-2$.
Our states are $(w_{m-2},p_{m-2})$ and $(w_{f(m-2)},p_{m-2})$. 
By Lemma \ref{lem:suffixword_equal}, we have $\delta_\cA(w_{m-2},a) = \delta_\cA(w_{f(m-2)},a)$ for all $a \in \Sig$. Thus non-empty words cannot distinguish the states. But recall that $p_{m-2} = q$ is a non-final state, so the states we are trying to distinguish are both non-final, and thus the empty word does not distinguish the states either. So these states are indistinguishable.

Now, suppose $m-2-i > 0$, that is, $i < m-2$.
Assume that states $(w_{i+1},p_{i+1})$ and $(w_{f(i+1)},p_{i+1})$ are indistinguishable.
We want to show that $(w_i,p_i)$ and $(w_{f(i)},p_i)$ are indistinguishable.
Since $f(i) < i < m-2$, both states are non-final, and thus the empty word cannot distinguish them.
By Lemma \ref{lem:suffixword_equal}, if $a \ne a_{i+1}$. then $\delta_\cA(w_i,a) = \delta_\cA(w_{f(i)},a)$ for all $a \in \Sig$. So only words that start with $a_{i+1}$ can possibly distinguish the states.
But by Lemma \ref{lem:suffixword_next}, letter $a_{i+1}$ sends the states to $(w_{i+1},p_{i+1})$ and $(w_{f(i+1)},p_{i+1})$, which are indistinguishable by the induction hypothesis.
Thus the states cannot be distinguished. \qed
\end{proof}

This establishes an upper bound of $(m-1)n - (m-2)$ on the state complexity of $\Sig^* w \cap T$.
Next, we prove this bound is tight.
\begin{theorem}
Suppose $m \ge 3$ and $n \ge 2$. 
There exists a non-empty word $w$ and a language $T$, with $\kappa(\{w\}) \le m$ and $\kappa(T) \le n$, such that $\kappa(\Sig^* w \cap T) = (m-1)n - (m-2)$.
\end{theorem}

\begin{proof}
Let $\Sig = \{a,b\}$ and let $w = b^{m-2}$.
Let $\cA$ be the DFA for $\Sig^* w$.
Let $T$ be the language accepted by the DFA $\cD$ with state set $Q_n$, alphabet $\Sig$, initial state $0$, final state set $\{0,\dotsc,n-2\}$, and transformations $a \co (0,\dotsc,n-1)$ and $b \co \id$.

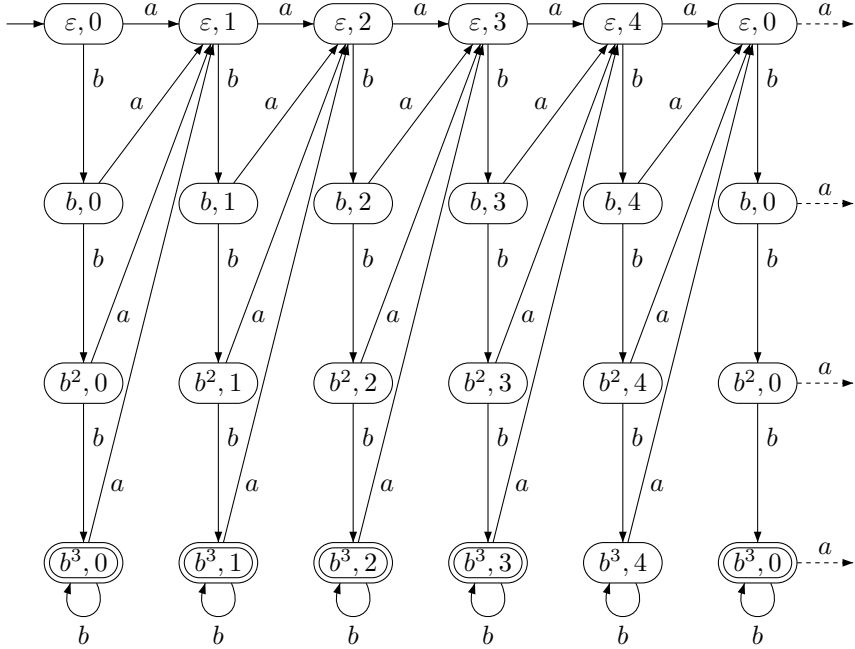
\begin{figure}[ht!]
\unitlength 8.5pt
\begin{center}\begin{picture}(36,28)(-2,-2)
\gasset{Nh=1.8,Nw=3.5,Nmr=1.25,ELdist=0.4,loopdiam=1.5}
\node(b00)(0 ,24){$\eps,0$}\imark(b00)
\node(b01)(6 ,24){$\eps,1$}
\node(b02)(12,24){$\eps,2$}
\node(b03)(18,24){$\eps,3$}
\node(b04)(24,24){$\eps,4$}
\node(b00')(30,24){$\eps,0$}
\node[Nframe=n](e0)(36,24){}
\node(b10)(0 ,16){$b,0$}
\node(b11)(6 ,16){$b,1$}
\node(b12)(12,16){$b,2$}
\node(b13)(18,16){$b,3$}
\node(b14)(24,16){$b,4$}
\node(b10')(30,16){$b,0$}
\node[Nframe=n](e1)(36,16){}
\node(b20)(0 ,8){$b^2,0$}
\node(b21)(6 ,8){$b^2,1$}
\node(b22)(12,8){$b^2,2$}
\node(b23)(18,8){$b^2,3$}
\node(b24)(24,8){$b^2,4$}
\node(b20')(30,8){$b^2,0$}
\node[Nframe=n](e2)(36,8){}
\node(b30)(0 ,0){$b^3,0$}\rmark(b30)
\node(b31)(6 ,0){$b^3,1$}\rmark(b31)
\node(b32)(12,0){$b^3,2$}\rmark(b32)
\node(b33)(18,0){$b^3,3$}\rmark(b33)
\node(b34)(24,0){$b^3,4$}
\node(b30')(30,0){$b^3,0$}\rmark(b30')
\node[Nframe=n](e3)(36,0){}
\drawedge[ELpos=30](b00,b10){$b$}
\drawedge[ELpos=30](b10,b20){$b$}
\drawedge[ELpos=30](b20,b30){$b$}
\drawloop[loopangle=-90](b30){$b$}
\drawedge[ELpos=30](b01,b11){$b$}
\drawedge[ELpos=30](b11,b21){$b$}
\drawedge[ELpos=30](b21,b31){$b$}
\drawloop[loopangle=-90](b31){$b$}
\drawedge[ELpos=30](b02,b12){$b$}
\drawedge[ELpos=30](b12,b22){$b$}
\drawedge[ELpos=30](b22,b32){$b$}
\drawloop[loopangle=-90](b32){$b$}
\drawedge[ELpos=30](b03,b13){$b$}
\drawedge[ELpos=30](b13,b23){$b$}
\drawedge[ELpos=30](b23,b33){$b$}
\drawloop[loopangle=-90](b33){$b$}
\drawedge[ELpos=30](b04,b14){$b$}
\drawedge[ELpos=30](b14,b24){$b$}
\drawedge[ELpos=30](b24,b34){$b$}
\drawloop[loopangle=-90](b34){$b$}
\drawedge[ELpos=30](b00',b10'){$b$}
\drawedge[ELpos=30](b10',b20'){$b$}
\drawedge[ELpos=30](b20',b30'){$b$}
\drawloop[loopangle=-90](b30'){$b$}
\drawedge(b00,b01){$a$}
\drawedge(b01,b02){$a$}
\drawedge(b02,b03){$a$}
\drawedge(b03,b04){$a$}
\drawedge(b04,b00'){$a$}
\drawedge[dash={0.2}0](b00',e0){$a$}
%\drawedge[curvedepth=-2.5,ELpos=40,ELdist=-1](b04,b00){$a$}
\drawedge(b10,b01){$a$}
\drawedge(b11,b02){$a$}
\drawedge(b12,b03){$a$}
\drawedge(b13,b04){$a$}
\drawedge(b14,b00'){$a$}
\drawedge[dash={0.2}0](b10',e1){$a$}
%\drawedge[curvedepth=-2.5,ELpos=40,ELdist=-1](b14,b10){$a$}
\drawedge[ELpos=20,ELdist=-1](b20,b01){$a$}
\drawedge[ELpos=20,ELdist=-1](b21,b02){$a$}
\drawedge[ELpos=20,ELdist=-1](b22,b03){$a$}
\drawedge[ELpos=20,ELdist=-1](b23,b04){$a$}
\drawedge[ELpos=20,ELdist=-1](b24,b00'){$a$}
\drawedge[dash={0.2}0](b20',e2){$a$}
%\drawedge[curvedepth=-2.5,ELpos=40,ELdist=-1](b24,b20){$a$}
\drawedge[ELpos=15,ELdist=-1](b30,b01){$a$}
\drawedge[ELpos=15,ELdist=-1](b31,b02){$a$}
\drawedge[ELpos=15,ELdist=-1](b32,b03){$a$}
\drawedge[ELpos=15,ELdist=-1](b33,b04){$a$}
\drawedge[ELpos=15,ELdist=-1](b34,b00'){$a$}
\drawedge[dash={0.2}0](b30',e3){$a$}
%\drawedge[curvedepth=-2.5,ELpos=40,ELdist=-1](b34,b30){$a$}

%\drawloop(n-1){$\Sig \setminus \{b\}$}
%\drawedge[curvedepth=4, ELdist=-1.1](n-1,0){$b$}
\end{picture}\end{center}
\caption{DFA $\cA \times \cD$ for matching a single suffix, with $m = 5$ and $n = 5$. Column 0 is duplicated to make the diagram cleaner; the actual DFA 
contains only
one copy of this column.}
\end{figure}

We show that $\cA \times \cD$ has $(m-1)n - (m-2)$ reachable and pairwise distinguishable states.
For reachability, for $0 \le i \le m-2$ and $0 \le q \le n-1$, we can reach $(b^i,q)$ from the initial state $(\eps,0)$ by the word $a^qb^i$.
For distinguishability, note that all $m-1$ states in column $n-1$ are indistinguishable, and so collapse to one state under the indistinguishability relation.
Indeed, given states $(b^i,n-1)$ and $(b^j,n-1)$, if we apply $a$ both states are sent to $(\eps,0)$, and if we apply $b$ we simply reach another pair of non-final states in column $n-1$.
Hence at most $(m-1)n - (m-2)$ of the reachable states are pairwise distinguishable.
Next consider $(b^i,q)$ and $(b^j,q)$ with $i < j$ and $q \ne n-1$.
We can distinguish these states by $b^{m-2-j}$.
So pairs of states in the same column are distinguishable, with the exception of states in column $n-1$.
For pairs of states in different columns, consider $(b^i,p)$ and $(b^j,q)$ with $p < q$.
If $q \ne n-1$, then by $a^{n-1-q}$ we reach $(\eps,n-1+p-q)$ and $(\eps,n-1)$. These latter states are distinguished by $w = b^{m-2}$.
If $q = n-1$, then $(b^i,p)$ and $(b^j,n-1)$ are distinguished by $b^{m-2-i}$.
Hence there are $(m-1)n - (m-2)$ reachable and pairwise distinguishable states. \qed
\end{proof}

\subsection{Matching a Single Factor}

\begin{proposition}
\label{prop:factorword}
If the state complexity of $\{w\}$ is $m$, then the state complexity of $\Sig^* w \Sig^*$ is $m-1$.
\end{proposition}
\begin{proof}
Let $\cA =( W, \Sig, \delta_\cA, w_0, \{w_{m-2}\} )$ be the DFA with transitions defined as follows: for all $a \in \Sig$ and $w_i \in W$, we have $w_ia \sto \delta_\cA(w_i,a) \sto w$.
Recall from Proposition \ref{prop:suffixword} that $\cA$ recognizes $\Sig^* w$.
We modify $\cA$ to obtain a DFA $\cA'$ that accepts $\Sig^* w \Sig^*$ as follows.
Let $\cA' =( W, \Sig, \delta_{\cA'}, w_0, \{w_{m-2}\} )$, where $\delta_{\cA'}$ is defined as follows for each $a \in \Sig$:
$\delta_{\cA'}(w_i,a) = \delta_\cA(w_i,a)$ for $i < m-2$, and $\delta_{\cA'}(w_{m-2},a) = w_{m-2}$.
Note that $\cA'$ is minimal: state $w_i$ can be reached by the word $w_i$, and states $w_i$ and $w_j$ with $i < j$ are distinguished by $a_{j+1} \dotsb a_{m-2}$. It remains to show that $\cA'$ accepts $\Sig^*w\Sig^*$.

To simplify the notation, we write $\delta'$ instead of $\delta_{\cA'}$ and $\delta$ instead of $\delta_\cA$.
Suppose $x$ is accepted by $\cA'$. 
Write $x = yz$, where $y$ is the shortest prefix of $x$ such that $\delta'(\eps,y) = w_{m-2}$. 
Since $y$ is minimal in length, for every proper prefix $y'$ of $y$, we have $\delta'(\eps,y') = w_i$ for some $i < m-2$.
It follows that $\delta'(\eps,y) = \delta(\eps,y)$  by the definition of $\delta'$.
So $\delta(\eps,y) = w_{m-2}$, and hence $y$ is accepted by $\cA$.
It follows that $y \in \Sig^* w$. 
This implies $x = yz \in \Sig^* w \Sig^*$.

Conversely, suppose $x \in \Sig^* w \Sig^*$. Write $x = ywz$ with $y$ minimal.
Since $yw \in \Sig^*w$, we have $\delta(\eps,yw) = w_{m-2}$.
Furthermore, $yw$ is the shortest prefix of $x$ such that $\delta(\eps,yw) = w_{m-2}$, since if there was a shorter prefix then $y$ would not be minimal.
This means that $\delta(\eps,yw) = \delta'(\eps,yw)$ by the definition of $\delta'$.
So $\delta'(\eps,ywz) = w_{m-2}$ and hence $x = ywz$ is accepted by $\cA'$. \qed
\end{proof}

Fix $w$ with state complexity $m$, and let $\cA$ and $\cA'$ be the DFAs for $\Sig^*w$ and $\Sig^* w \Sig^*$, respectively, as described in the proof of Proposition \ref{prop:factorword}. Fix $T$ with state complexity 
at most
 $n$, and let $\cD$ be an $n$-state DFA for $T$ with state set $Q_n$ and final state set $F$. The direct product DFA $\cA' \times \cD$ with final state set $\{w\} \times F$ recognizes $\Sig^* w \Sig^* \cap T$. Since $\cA' \times \cD$ has $(m-1)n$ states, this gives an upper bound of $(m-1)n$ on the state complexity of $\Sig^* w \Sig^* \cap T$. We claim that this upper bound is tight.

\begin{theorem}
Suppose $m \ge 3$ and $n \ge 2$. 
There exists a non-empty word $w$ and a language $T$, with $\kappa(\{w\}) \le m$ and $\kappa(T) \le n$, such that $\kappa(\Sig^* w\Sig^* \cap T) = (m-1)n$.
\end{theorem}
\begin{proof}
Let $\Sig = \{a,b\}$ and let $w = b^{m-2}$.
Let $\cA'$ be the DFA for $\Sig^* w \Sig^*$.
Let $T$ be the language accepted by the DFA $\cD$ with state set $Q_n$, alphabet $\Sig$, initial state $0$, final state set $\{0,\dotsc,n-2\}$, and transformations $a \co (0,\dotsc,n-1)$ and $b \co \id$.

\begin{figure}[ht!]
\unitlength 8.5pt
\begin{center}\begin{picture}(36,28)(-2,-2)
\gasset{Nh=1.8,Nw=3.5,Nmr=1.25,ELdist=0.4,loopdiam=1.5}
\node(b00)(0 ,24){$\eps,0$}\imark(b00)
\node(b01)(6 ,24){$\eps,1$}
\node(b02)(12,24){$\eps,2$}
\node(b03)(18,24){$\eps,3$}
\node(b04)(24,24){$\eps,4$}
\node(b00')(30,24){$\eps,0$}
\node[Nframe=n](e0)(36,24){}
\node(b10)(0 ,16){$b,0$}
\node(b11)(6 ,16){$b,1$}
\node(b12)(12,16){$b,2$}
\node(b13)(18,16){$b,3$}
\node(b14)(24,16){$b,4$}
\node(b10')(30,16){$b,0$}
\node[Nframe=n](e1)(36,16){}
\node(b20)(0 ,8){$b^2,0$}
\node(b21)(6 ,8){$b^2,1$}
\node(b22)(12,8){$b^2,2$}
\node(b23)(18,8){$b^2,3$}
\node(b24)(24,8){$b^2,4$}
\node(b20')(30,8){$b^2,0$}
\node[Nframe=n](e2)(36,8){}
\node(b30)(0 ,0){$b^3,0$}\rmark(b30)
\node(b31)(6 ,0){$b^3,1$}\rmark(b31)
\node(b32)(12,0){$b^3,2$}\rmark(b32)
\node(b33)(18,0){$b^3,3$}\rmark(b33)
\node(b34)(24,0){$b^3,4$}
\node(b30')(30,0){$b^3,0$}\rmark(b30')
\node[Nframe=n](e3)(36,0){}
\drawedge[ELpos=30](b00,b10){$b$}
\drawedge[ELpos=30](b10,b20){$b$}
\drawedge(b20,b30){$b$}
\drawloop[loopangle=-90](b30){$b$}
\drawedge[ELpos=30](b01,b11){$b$}
\drawedge[ELpos=30](b11,b21){$b$}
\drawedge(b21,b31){$b$}
\drawloop[loopangle=-90](b31){$b$}
\drawedge[ELpos=30](b02,b12){$b$}
\drawedge[ELpos=30](b12,b22){$b$}
\drawedge(b22,b32){$b$}
\drawloop[loopangle=-90](b32){$b$}
\drawedge[ELpos=30](b03,b13){$b$}
\drawedge[ELpos=30](b13,b23){$b$}
\drawedge(b23,b33){$b$}
\drawloop[loopangle=-90](b33){$b$}
\drawedge[ELpos=30](b04,b14){$b$}
\drawedge[ELpos=30](b14,b24){$b$}
\drawedge(b24,b34){$b$}
\drawloop[loopangle=-90](b34){$b$}
\drawedge[ELpos=30](b00',b10'){$b$}
\drawedge[ELpos=30](b10',b20'){$b$}
\drawedge(b20',b30'){$b$}
\drawloop[loopangle=-90](b30'){$b$}
\drawedge(b00,b01){$a$}
\drawedge(b01,b02){$a$}
\drawedge(b02,b03){$a$}
\drawedge(b03,b04){$a$}
\drawedge(b04,b00'){$a$}
\drawedge[dash={0.2}0](b00',e0){$a$}
%\drawedge[curvedepth=-2.5,ELpos=40,ELdist=-1](b04,b00){$a$}
\drawedge(b10,b01){$a$}
\drawedge(b11,b02){$a$}
\drawedge(b12,b03){$a$}
\drawedge(b13,b04){$a$}
\drawedge(b14,b00'){$a$}
\drawedge[dash={0.2}0](b10',e1){$a$}
%\drawedge[curvedepth=-2.5,ELpos=40,ELdist=-1](b14,b10){$a$}
\drawedge[ELpos=20,ELdist=-1](b20,b01){$a$}
\drawedge[ELpos=20,ELdist=-1](b21,b02){$a$}
\drawedge[ELpos=20,ELdist=-1](b22,b03){$a$}
\drawedge[ELpos=20,ELdist=-1](b23,b04){$a$}
\drawedge[ELpos=20,ELdist=-1](b24,b00'){$a$}
\drawedge[dash={0.2}0](b20',e2){$a$}
%\drawedge[curvedepth=-2.5,ELpos=40,ELdist=-1](b24,b20){$a$}
\drawedge(b30,b31){$a$}
\drawedge(b31,b32){$a$}
\drawedge(b32,b33){$a$}
\drawedge(b33,b34){$a$}
\drawedge(b34,b30'){$a$}
\drawedge[dash={0.2}0](b30',e3){$a$}

%\drawloop(n-1){$\Sig \setminus \{b\}$}
%\drawedge[curvedepth=4, ELdist=-1.1](n-1,0){$b$}
\end{picture}\end{center}
\caption{DFA $\cA' \times \cD$ for matching a single factor, with $m = 5$ and $n = 5$. Column 0 is duplicated to make the diagram cleaner; the actual DFA 
contains only
one copy of this column.}
\end{figure}
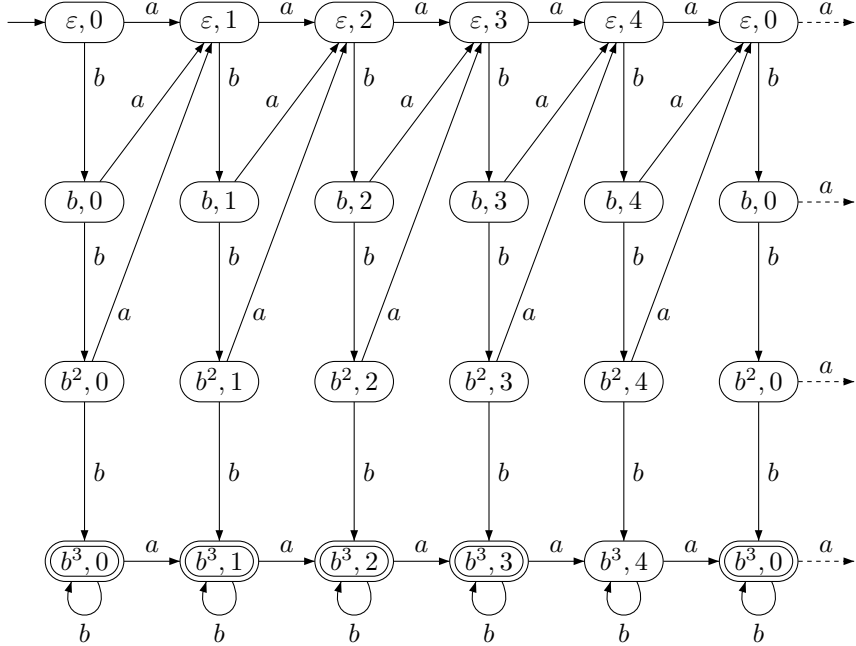

We show that $\cA' \times \cD$ has $(m-1)n$ reachable and pairwise distinguishable states.
For reachability, for $0 \le i \le m-2$ and $0 \le q \le n-1$, we can reach $(b^i,q)$ from the initial state $(\eps,0)$ by the word $a^qb^i$.
For distinguishability,
suppose we have states $(b^i,q)$ and $(b^j,q)$ in the same column $q$, with $i < j$.
By $b^{m-2-j}$ we reach $(b^{m-2+i-j},q)$ and $(w,q)$, with $b^{m-2+i-j} \ne w$.
Then by $a$ we reach $(\eps,qa)$ and $(w,qa)$, which are distinguishable by a word in $a^*$.
For states in different columns, suppose we have $(b^i,p)$ and $(b^j,q)$ with $p < q$.
By a sufficiently long word in $b^*$, we reach $(w,p)$ and $(w,q)$.
These states are distinguishable by $a^{n-1-q}$.
So all reachable states are pairwise distinguishable. \qed
\end{proof}

\subsection{Matching a Single Subsequence}

\begin{proposition}
\label{prop:subseqword}
If the state complexity of $\{w\}$ is $m$, then the state complexity of $\Sig^* \shu w$ is $m-1$.
\end{proposition}
\begin{proof}
Define a DFA $\cA = (W,\Sig,\delta_\cA,\eps,\{w\})$ where $\delta_\cA(w_i,a_{i+1}) = w_{i+1}$, and $\delta_\cA(w_i,a) = w_i$ for $a \ne a_{i+1}$. Note that $\cA$ is minimal: state $w_i$ is reached by word $w_i$ and states $w_i,w_j$ with $i < j$ are distinguished by $a_{j+1} \dotsb a_{m-2}$. We claim that $\cA$ recognizes $\Sig^* \shu w$.

Write $\delta$ rather than $\delta_\cA$ to simplify the notation.
Suppose $x \in \Sig^* \shu w$. Then we can write $x = x_0a_1x_1a_2 x_2 \dotsb a_{m-2} x_{m-2}$, where $x_0,\dotsc,x_{m-2} \in \Sig^*$. We claim that $\delta(\eps,x_0a_1x_1\dotsb a_ix_i) = w_j$ for some $j \ge i$. We proceed by induction on $i$. The base case $i = 0$ is trivial.

Now, suppose that $i > 0$ and 
$\delta(\eps,x_0a_1x_1 \dotsb a_{i-1}x_{i-1}) = w_j$
for some $j \ge i-1$.
Then
$\delta(\eps,x_0a_1x_1 \dotsb a_{i}x_{i}) = \delta(w_j,a_ix_i)$.
We consider two cases:
\bi
\item
If $j = i-1$, we have $\delta(w_{i-1},a_ix_i) = \delta(w_i,x_i) = w_k$ for some $k$ with $k \ge i$, as required.
\item
If $j > i-1$, we have $\delta(w_i,a_ix_i) = w_k$  for some $k$ with $k \ge i$, as required.
\ei
This completes the inductive proof. It follows then that 
$\delta(\eps,x) = w_{m-2} = w$, and so $x$ is accepted by $\cA$.
Conversely, if $x$ is accepted by $\cA$, then it is clear from the definition of the transition function that the letters $a_1,a_2,\dotsc,a_{m-2}$ must occur within $x$ in order, and so $x \in \Sig^* \shu w$. \qed
\end{proof}

Fix $w$ with state complexity $m$, and let $\cA$ be the DFA for $\Sig^* \shu w$ described in the proof of Proposition \ref{prop:subseqword}. Fix $T$ with state complexity 
at most
$n$, and let $\cD$ be an $n$-state DFA for $T$ with state set $Q_n$ and final state set $F$. The direct product DFA $\cA \times \cD$ with final state set $\{w\} \times F$ recognizes $(\Sig^* \shu w)\cap T$. Since $\cA \times \cD$ has $(m-1)n$ states, this gives an upper bound of $(m-1)n$ on the state complexity of $(\Sig^* \shu w) \cap T$. We claim that this upper bound is tight.

\begin{theorem}
Suppose $m \ge 3$ and $n \ge 2$. 
There exists a non-empty word $w$ and a language $T$, with $\kappa(\{w\}) \le m$ and $\kappa(T) \le n$, such that $\kappa((\Sig^* \shu w) \cap T) = (m-1)n$.
\end{theorem}
\begin{proof}
Let $\Sig = \{a,b\}$ and let $w = b^{m-2}$.
Let $\cA$ be the DFA for $\Sig^* \shu w$.
Let $T$ be the language accepted by the DFA $\cD$ with state set $Q_n$, alphabet $\Sig$, initial state $0$, final state set $\{0,\dotsc,n-2\}$, and transformations $a \co (0,\dotsc,n-1)$ and $b \co \id$.

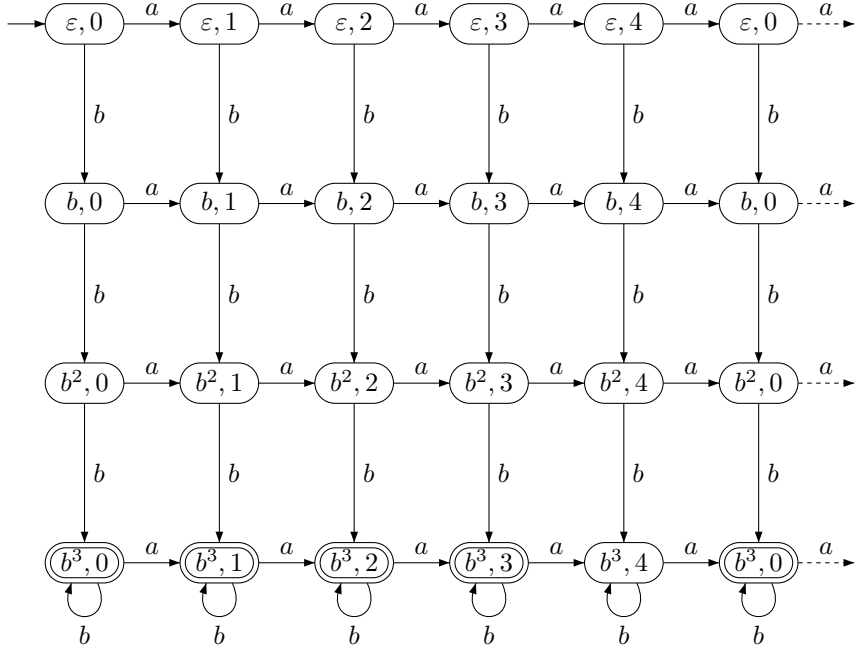
\begin{figure}[ht!]
\unitlength 8.5pt
\begin{center}\begin{picture}(36,28)(-2,-2)
\gasset{Nh=1.8,Nw=3.5,Nmr=1.25,ELdist=0.4,loopdiam=1.5}
\node(b00)(0 ,24){$\eps,0$}\imark(b00)
\node(b01)(6 ,24){$\eps,1$}
\node(b02)(12,24){$\eps,2$}
\node(b03)(18,24){$\eps,3$}
\node(b04)(24,24){$\eps,4$}
\node(b00')(30,24){$\eps,0$}
\node[Nframe=n](e0)(36,24){}
\node(b10)(0 ,16){$b,0$}
\node(b11)(6 ,16){$b,1$}
\node(b12)(12,16){$b,2$}
\node(b13)(18,16){$b,3$}
\node(b14)(24,16){$b,4$}
\node(b10')(30,16){$b,0$}
\node[Nframe=n](e1)(36,16){}
\node(b20)(0 ,8){$b^2,0$}
\node(b21)(6 ,8){$b^2,1$}
\node(b22)(12,8){$b^2,2$}
\node(b23)(18,8){$b^2,3$}
\node(b24)(24,8){$b^2,4$}
\node(b20')(30,8){$b^2,0$}
\node[Nframe=n](e2)(36,8){}
\node(b30)(0 ,0){$b^3,0$}\rmark(b30)
\node(b31)(6 ,0){$b^3,1$}\rmark(b31)
\node(b32)(12,0){$b^3,2$}\rmark(b32)
\node(b33)(18,0){$b^3,3$}\rmark(b33)
\node(b34)(24,0){$b^3,4$}
\node(b30')(30,0){$b^3,0$}\rmark(b30')
\node[Nframe=n](e3)(36,0){}
\drawedge[ELpos=50](b00,b10){$b$}
\drawedge[ELpos=50](b10,b20){$b$}
\drawedge[ELpos=50](b20,b30){$b$}
\drawloop[loopangle=-90](b30){$b$}
\drawedge[ELpos=50](b01,b11){$b$}
\drawedge[ELpos=50](b11,b21){$b$}
\drawedge[ELpos=50](b21,b31){$b$}
\drawloop[loopangle=-90](b31){$b$}
\drawedge[ELpos=50](b02,b12){$b$}
\drawedge[ELpos=50](b12,b22){$b$}
\drawedge[ELpos=50](b22,b32){$b$}
\drawloop[loopangle=-90](b32){$b$}
\drawedge[ELpos=50](b03,b13){$b$}
\drawedge[ELpos=50](b13,b23){$b$}
\drawedge[ELpos=50](b23,b33){$b$}
\drawloop[loopangle=-90](b33){$b$}
\drawedge[ELpos=50](b04,b14){$b$}
\drawedge[ELpos=50](b14,b24){$b$}
\drawedge[ELpos=50](b24,b34){$b$}
\drawloop[loopangle=-90](b34){$b$}
\drawedge[ELpos=50](b00',b10'){$b$}
\drawedge[ELpos=50](b10',b20'){$b$}
\drawedge[ELpos=50](b20',b30'){$b$}
\drawloop[loopangle=-90](b30'){$b$}
\drawedge(b00,b01){$a$}
\drawedge(b01,b02){$a$}
\drawedge(b02,b03){$a$}
\drawedge(b03,b04){$a$}
\drawedge(b04,b00'){$a$}
\drawedge[dash={0.2}0](b00',e0){$a$}
%\drawedge[curvedepth=-2.5,ELpos=40,ELdist=-1](b04,b00){$a$}
\drawedge(b10,b11){$a$}
\drawedge(b11,b12){$a$}
\drawedge(b12,b13){$a$}
\drawedge(b13,b14){$a$}
\drawedge(b14,b10'){$a$}
\drawedge[dash={0.2}0](b10',e1){$a$}
%\drawedge[curvedepth=-2.5,ELpos=40,ELdist=-1](b14,b10){$a$}
\drawedge(b20,b21){$a$}
\drawedge(b21,b22){$a$}
\drawedge(b22,b23){$a$}
\drawedge(b23,b24){$a$}
\drawedge(b24,b20'){$a$}
\drawedge[dash={0.2}0](b20',e2){$a$}
%\drawedge[curvedepth=-2.5,ELpos=40,ELdist=-1](b24,b20){$a$}
\drawedge(b30,b31){$a$}
\drawedge(b31,b32){$a$}
\drawedge(b32,b33){$a$}
\drawedge(b33,b34){$a$}
\drawedge(b34,b30'){$a$}
\drawedge[dash={0.2}0](b30',e3){$a$}
%\drawedge[curvedepth=-2.5,ELpos=40,ELdist=-1](b34,b30){$a$}

%\drawloop(n-1){$\Sig \setminus \{b\}$}
%\drawedge[curvedepth=4, ELdist=-1.1](n-1,0){$b$}
\end{picture}\end{center}
\caption{DFA $\cA \times \cD$ for matching a single subsequence, with $m = 5$ and $n = 5$. Column 0 is duplicated to make the diagram cleaner; the actual DFA 
contains only
 one copy of this column.}
\end{figure}

We show that $\cA \times \cD$ has $(m-1)n$ reachable and pairwise distinguishable states.
For reachability, for $0 \le i \le m-2$ and $0 \le q \le n-1$, we can reach $(b^i,q)$ from the initial state $(\eps,0)$ by the word $a^qb^i$.
For distinguishability,
suppose we have states $(b^i,q)$ and $(b^j,q)$ in the same column $q$, with $i < j$.
By $b^{m-2-j}$ we reach $(b^{m-2+i-j},q)$ and $(w,q)$, with $b^{m-2+i-j} \ne w$.
These states are distinguishable by  a word in $a^*$.
For states in different columns, suppose we have $(b^i,p)$ and $(b^j,q)$ with $p < q$.
By a sufficiently long word in $b^*$, we reach $(w,p)$ and $(w,q)$.
These states are distinguishable by $a^{n-1-q}$.
So all reachable states are pairwise distinguishable. \qed
\end{proof}

\section{Conclusions}
We investigated the state complexity of four new combined operations on regular languages, inspired by pattern matching problems, in both the general case and the case where the pattern set is a single word.
The operations we considered were of the form ``the intersection of $T$ with the right (left, two-sided, all-sided) ideal generated by $P$'', corresponding to searching for prefixes (suffixes, factors, subsequences) from a set of patterns $P$ in a set of texts $T$.
In the general case, the state complexity of these combined operations is just equal to the composition of the complexities of the individual operations; the complexity is polynomial in the case of prefix matching, and exponential (in the first parameter) in the case of suffix, factor and subsequence matching.
For single-word pattern sets the complexity is significantly lower: linear in the case of prefix matching, and polynomial in the other cases.
In all cases, the maximal complexity can 
be achieved only
by languages over an alphabet of at least two letters. 
For unary languages, the general case and single-word case coincide, and the four operations are all equivalent. 
The complexity is linear in the unary case.

\section*{References}
\bibliographystyle{model1b-num-names}
\bibliography{matching}
\end{document}